\documentclass[journal]{IEEEtran}
\usepackage{cite}
\usepackage{amsmath,amssymb,amsthm}
\usepackage[ruled,vlined,linesnumbered]{algorithm2e}
\usepackage{comment}
\usepackage{url}

\newtheorem{theorem}{Theorem}

\newtheorem{corollary}{Corollary}

\newtheorem{lemma}{Lemma}

\usepackage{tikz}
\usepackage{xcolor}
\usetikzlibrary{positioning}
\usepackage{pgfplots}
\usetikzlibrary{decorations.pathreplacing}

\DeclareMathOperator{\supp}{supp}

\definecolor{color1}{rgb}{0.0000,0.4470,0.7410}
\definecolor{color2}{rgb}{0.8500,0.3250,0.0980}
\definecolor{color3}{rgb}{0.9290,0.6940,0.1250}
\definecolor{color4}{rgb}{0.4940,0.1840,0.5560}
\definecolor{color5}{rgb}{0.4660,0.6740,0.1880}
\definecolor{color6}{rgb}{0.3010,0.7450,0.9330}
\definecolor{color7}{rgb}{0.6350,0.0780,0.1840}

\definecolor{myred}{RGB}{222,45,38}
\definecolor{myorange}{RGB}{254,178,76}
\definecolor{myyellow}{RGB}{255,237,160}
\definecolor{mygreen}{RGB}{49,163,84}
\definecolor{myblue}{RGB}{49,130,189}
\definecolor{mypurple}{RGB}{197,27,138}

\pgfplotsset{compat=1.17}

\begin{document}
\title{Erasure Decoding of Quantum LDPC Codes \\ via Peeling with Cluster Decomposition}
\title{Improved Peeling Decoding of Quantum LDPC Codes over Erasures with Cluster Decomposition}
\title{Cluster Decomposition for Improved Erasure Decoding of Quantum LDPC Codes}

\author{Hanwen~Yao, 
        Mert~G\"okduman, 
        and~Henry~D.~Pfister 
\thanks{Hanwen~Yao is with Duke University, Durham, NC 27708 USA 
(e-mail: hanwen.yao@duke.edu).}
\thanks{Mert~G\"okduman is with Duke University, Durham, NC 27708 USA 
(e-mail: mert.gokduman@duke.edu).}
\thanks{Henry~D.~Pfister is with Duke University, Durham, NC 27708 USA 
(e-mail: henry.pfister@duke.edu).}
}

\maketitle

\begin{abstract}
We introduce a new erasure decoder that applies to arbitrary quantum LDPC codes.
Dubbed the cluster decoder, it generalizes the decomposition idea of Vertical-Horizontal (VH) decoding introduced by Connelly et al.\ in 2022.
Like the VH decoder, the idea is to first run the peeling decoder and then post-process the resulting stopping set.
The cluster decoder breaks the stopping set into a tree of clusters which can be solved  sequentially via Gaussian Elimination (GE).
By allowing clusters of unconstrained size, this decoder achieves maximum-likelihood (ML) performance with reduced complexity compared with full GE. 
When GE is applied only to clusters whose sizes are less than a constant, the performance is degraded but the complexity becomes linear in the block length.
Our simulation results show that, for hypergraph product codes, the cluster decoder with constant cluster size achieves near-ML performance similar to VH decoding in the low-erasure-rate regime. 
For the general quantum LDPC codes we studied, 
the cluster decoder can be used to estimate the ML performance curve with reduced complexity over a wide range of erasure rates.
\end{abstract}

\section{Introduction}
\label{sec:intro}
\IEEEPARstart{T}{o} achieve scalable and fault-tolerant quantum computation, error correction is essential for protecting quantum information against noise.
Among various error correction schemes, quantum low-density parity-check (LDPC) codes, 
such as hypergraph product codes \cite{tillich2013quantum}, 
stand out as promising candidates as they offer a much lower 
overhead \cite{gottesman2013fault,fawzi2020constant} 
compared with topological codes such as surface codes 
\cite{kitaev2003fault,dennis2002topological} and 
color codes \cite{bombin2006topological}.
Recent breakthrough results have shown constructions of quantum LDPC codes that are asymptotically \emph{good}, 
achieving both constant rate and linear minimum distance \cite{panteleev2022asymptotically,leverrier2022quantum,dinur2023good}. 
From the standpoint of practical implementation, Bravyi et al. \cite{bravyi2024high} have also constructed some bivariate bicycle codes that can be laid out on a two-dimensional grid with 
some transversal edges between them, 
which is well suited for architectures based on superconducting qubits.

In this work, we focus on decoding quantum LDPC codes 
over the quantum erasure channel \cite{bennett1997capacities}, 
a model in which the encoded qubits are affected by uniform random Pauli errors 
with locations known by the decoder. 
Attention has been drawn to this error model recently by 
new proposals and demonstrations of erased qubits in various architectures, including neutral atoms 
\cite{wu2022erasure,sahay2023high,ma2023high}, 
trapped ions \cite{kang2023quantum}, 
and superconducting qubits \cite{kubica2023erasure,teoh2023dual}. 
Quantum codes, mostly surface codes, have also been shown to achieve significantly higher thresholds over erasures compared to the standard Pauli noise in the code capacity model, the phenomenological noise model, and circuit-level noise simulations \cite{wu2022erasure,sahay2023high,kang2023quantum,kubica2023erasure,stace2009thresholds,barrett2010fault,whiteside2014upper}.

Several decoding algorithms have been studied for 
quantum erasure correction in the code capacity model, 
targeting different code families. 
Most of them build upon the classical peeling decoder 
\cite{luby2001efficient,richardson2001efficient}, 
an iterative erasure decoding algorithm that runs on the Tanner graph. 
In \cite{delfosse2020linear}, a linear-time decoder was introduced for surface codes, achieving maximum-likelihood (ML) decoding performance by peeling on a spanning tree of erasures on the surface code lattice. 
This was later extended into the union-find decoder, which can correct both Pauli errors and erasures with a higher complexity, 
initially for topological codes \cite{delfosse2021almost}, 
and later generalized to quantum LDPC codes \cite{delfosse2022toward}. 
In \cite{lee2020trimming}, a trimming decoder was proposed for the erasure decoding of color codes that combines peeling on a spanning tree with erasure set extension or vertex inactivation.
Erasure decoding of subsystem color codes has also been studied with a combination of techniques, including peeling, clustering, and gauge fixing \cite{solanki2021correcting,solanki2023decoding}.

In 2022, two erasure decoders were proposed by Connolly \emph{et al.}: 
the pruned peeling decoder and the Vertical-Horizontal (VH) decoder \cite{connolly2024fast}. 
The pruned peeling decoder combines the classical peeling decoder with a search for the stabilizers contained entirely within erasures.
The VH decoder is designed specifically for hypergraph product codes, 
and it integrates pruned peeling with iteratively 
addressing vertical and horizontal stopping sets 
following the structure of hypergraph product construction. 
For the hypergraph product codes simulated in \cite{connolly2024fast}, 
VH decoder demonstrated near-ML performance 
in the low-erasure-rate regime with a computational complexity of $O(n^2)$ for codes of length $n$.

Another recent work by G\"okduman \emph{et al.} proposed to apply the Belief Propagation with Guided Decimation (BPGD) algorithm as a general-purpose erasure decoder for quantum LDPC codes \cite{gokduman2024erasure}. 
Previously, BPGD has been shown to offer competitive performance for decoding both bit-flip and depolarizing noise in the code capacity model \cite{yao2024belief}. 
With natural modifications that incorporate Log Likelihood Ratio (LLR) adjustments and damping, BPGD on the hypergraph product codes exhibited close-to-ML performance similar to VH in the low-erasure-rate regime \cite{gokduman2024erasure}.  
However, for erasure rates near the threshold, a performance gap remains between BPGD and ML for the quantum LDPC codes studied in \cite{gokduman2024erasure}.

In this work, we introduce a new erasure decoding algorithm for quantum LDPC codes, called the cluster decoder, which integrates classical peeling with a post-processing step termed cluster decomposition. 
When the peeling decoder fails to fully recover the erasures, we decompose the residual Tanner graph induced by the remaining stopping set into a forest of clusters, with each cluster corresponding to a biconnected component of the residual Tanner graph. 
We then show that the erasures within these clusters can be solved sequentially and subsequently combined to yield a global solution for the stopping set.

To resolve the erasures remaining after peeling,  cluster decomposition breaks the problem into a series of subproblems, each associated with a distinct cluster, and solves them sequentially. 
Unlike the VH decomposition, which relies on the specific structure of hypergraph product codes, our approach is independent of the code structure and applies to arbitrary quantum LDPC codes.

By allowing clusters of unconstrained sizes, our cluster decoder achieves ML decoding performance with reduced complexity compared to direct Gaussian Elimination. 
When we constrain the cluster sizes by a constant, the cluster decoder can be shown to achieve linear decoding complexity. 

Our simulation results show that, for hypergraph product codes, the cluster decoder with constrained cluster sizes, similar to the VH decoder, closely approaches the ML performance in the low-erasure-rate regime. 
For the general quantum LDPC code we tested, with reduced complexity, the cluster decoder effectively produces the ML performance curve in the high-erasure-rate regime.

\section{Preliminaries}
\label{sec:prelim}
\subsection{Graphs}
\label{subsec:graphs}
A \emph{graph} $\mathcal{G}=(\mathcal{V},\mathcal{E})$ 
is defined by a node set $\mathcal{V}$ and an edge set $\mathcal{E}$ 
consisting of edges $e$ represented 
by unordered pairs $e = (v_1,v_2)$ for $v_1,v_2\in \mathcal{V}$. 
All graphs discussed in this paper are assumed to be 
{undirected} without {loops} and {multi-edges}.
A graph $\mathcal{G}$ is called \emph{connected} if there is a {path} 
between every pair of vertices. 
Otherwise, $\mathcal{G}$ is called \emph{disconnected}.

Given a graph $\mathcal{G} = (\mathcal{V},\mathcal{E})$, we define:
\begin{itemize}
    \item {\bf connected component}: a maximal set of nodes in 
    $\mathcal{V}$ whose induced subgraph is connected.
    \item {\bf articulation point}: a node $v\in \mathcal{V}$ 
    whose removal with its adjacent edges increases 
    the number of connected components in $\mathcal{G}$.
    \item {\bf bridge}: an edge $e\in \mathcal{E}$ 
    whose removal increases the number of connected components 
    in $\mathcal{G}$.
    \item {\bf biconnected component}: a maximal set of nodes in 
    $\mathcal{V}$ whose induced subgraph has no articulation point.
\end{itemize}
At the top of Figure~\ref{fig:bc-tree} we see an example of a graph 
showing the above-defined elements. 
The graph $\mathcal{G}$ in this example is connected by itself, 
so it only has one connected component. 
$\mathcal{G}$ has two articulation points node 5 and node 8, 
and two bridges edge $(8,10)$ and edge $(8,11)$.
$\mathcal{G}$ has five biconnected components 
$\{3,4,5,7\},\{1,2,5\},\{5,6,8,9\},\{8,11\}$, and $\{8,10\}$.
Note that the set of two nodes involved in a bridge 
also counts as a biconnected component.

The biconnected components and articulation points of a graph 
can be computed with the classic Hopcroft-Tarjan's algorithm 
\cite{hopcroft1973algorithm}. 
This sequential algorithm works by performing a depth-first search 
on $\mathcal{G}$. 
For each connected component, edges not included 
in the depth-first search tree are called \emph{back edges} 
if they connect a node to its ancestor in the tree 
\cite[page~569]{cormen2022introduction}. 
A non-root node $s$ is an articulation point if the subtree 
rooted at $s$ has no back edges pointing to any ancestor of $s$, 
and a root node $s$ is an articulation point if it has more than one 
descendant in the depth-first search tree.
For a graph $\mathcal{G}=(\mathcal{V},\mathcal{E})$, 
the Hopcroft-Tarjan's algorithm runs with time complexity 
$O(|\mathcal{V}|+|\mathcal{E}|)$ \cite{hopcroft1973algorithm}.

Any connected graph can be decomposed into a tree of biconnected components and articulation points 
\cite{Harary1967,gallai1964elementare}.
In this work, we call it the \emph{cluster tree} and refer to each of the biconnected components as a \emph{cluster}.
Figure~\ref{fig:bc-tree} shows an example of this decomposition.
A cluster tree of a graph $G$ has two types of nodes: \emph{cut nodes} representing the articulation points in $G$, and \emph{cluster nodes} representing the biconnected components in $G$.
A graph with more than one connected component can be decomposed into a cluster forest. 

The following lemma shows that the cluster forest of a graph $G$ can be constructed in linear time using Hopcroft-Tarjan's algorithm.
\begin{lemma}
\label{lm:cluster_construction_complexity}
    For a graph $\mathcal{G}=(\mathcal{V},\mathcal{E})$, 
    the construction of the cluster forest can be achieved with 
    time complexity at most $O(2|\mathcal{V}|+|\mathcal{E}|)$.
\end{lemma}
\begin{proof}
To construct the cluster forest for a graph 
$\mathcal{G}=(\mathcal{V},\mathcal{E})$,
we can first apply Hopcroft-Tarjan's algorithm 
with time complexity $O(|\mathcal{V}|+|\mathcal{E}|)$ to
obtain the set of clusters $B_1,\ldots,B_s$ 
and the set of cut nodes $s_1,\ldots,s_t$.
Then, for all $i,j$ where $s_j$ is contained in $B_i$, we add an edge between the cluster $B_i$ 
and the cut node $s_j$.
This can be achieved by examining each node within every cluster, 
which takes a time complexity of $O(|\mathcal{V}|)$.
Therefore, the cluster forest can be constructed with time complexity
at most $O(2|\mathcal{V}|+|\mathcal{E}|)$.
\end{proof}

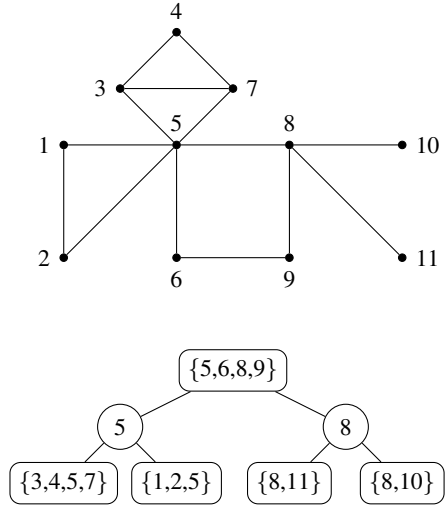
\begin{figure}[t]
    \centering
    \begin{tikzpicture}[scale=1.5,every label/.append style={font=\small}]
        \tikzset{
            blacknode/.style={draw,circle,fill=black,inner sep=1pt,font=\tiny},
            blocknode/.style={draw,rectangle,rounded corners,font=\small},
            cutnode/.style={draw,circle,font=\small},
        }
        \node[blacknode,label=left:{1}] (1) at (0,1) {};
        \node[blacknode,label=left:{2}] (2) at (0,0) {};
        \node[blacknode,label=left:{3}] (3) at (0.5,1.5) {};
        \node[blacknode,label=above:{4}] (4) at (1,2) {};
        \node[blacknode,label=above:{5}] (5) at (1,1) {};
        \node[blacknode,label=below:{6}] (6) at (1,0) {};
        \node[blacknode,label=right:{7}] (7) at (1.5,1.5) {};
        \node[blacknode,label=above:{8}] (8) at (2,1) {};
        \node[blacknode,label=below:{9}] (9) at (2,0) {};
        \node[blacknode,label=right:{10}] (10) at (3,1) {};
        \node[blacknode,label=right:{11}] (11) at (3,0) {};
        \draw (1) -- (2) -- (5) -- (1);
        \draw (3) -- (4) -- (7) -- (3) -- (5) -- (7);
        \draw (5) -- (6) -- (9) -- (8) -- (5);
        \draw (8) -- (10);
        \draw (8) -- (11);

        \begin{scope}[shift={(1.5,-1)}]
        \node[blocknode] (6) at (0, 0) {\{5,6,8,9\}};
        \node[blocknode] (10) at (1.5,-1) {\{8,10\}};
        \node[blocknode] (11) at (0.5,-1) {\{8,11\}};
        \node[blocknode] (3) at (-1.5,-1) {\{3,4,5,7\}};
        \node[blocknode] (1) at (-0.5,-1) {\{1,2,5\}};
        \node[cutnode] (8) at (1,-0.5) {8};
        \node[cutnode] (5) at (-1,-0.5) {5};
        \draw (5) -- (6) -- (8);
        \draw (8) -- (10);
        \draw (8) -- (11);
        \draw (3) -- (5);
        \draw (1) -- (5);
        \end{scope}
    \end{tikzpicture}
    \caption{Example of decomposing a graph $\mathcal{G}$ 
    into a cluster tree.}
    \label{fig:bc-tree}
\end{figure}

\subsection{Binary Linear Codes and Erasure Syndrome Decoding}
Let $\mathbb{F}_2 = \{0,1\}$ be the binary Galois field 
defined by modulo-2 addition and multiplication.
A length-$n$ binary linear code $\mathcal{C}\subseteq\mathbb{F}_2^n$ 
is a subset of length-$n$ binary strings satisfying 
$w+x \in \mathcal{C}$ for all $w,x \in \mathcal{C}$.
Such a code forms a vector space over $\mathbb{F}_2$.
A generator matrix $G\in \mathbb{F}_2^{k \times n}$ for 
$\mathcal{C}$ is a $k \times n$ matrix whose rows span the code.
A parity-check matrix $H\in \mathbb{F}_2^{(n-k) \times n}$ 
for $\mathcal{C}$ is an $(n-k) \times n$ matrix 
whose rows are orthogonal to the code.

For transmission over the binary erasure channel (BEC), 
each code bit in a codeword $x\in \mathcal{C}$ 
is erased with some probability, 
with the received vector given by $y\in\{0,1,e\}^n$, 
where $e$ indicates an \emph{erasure}.
In \textit{erasure syndrome decoding}, the decoder only sees the locations of the erasures and the syndrome \(s = H\tilde{y}\), where $\tilde {y} \in \mathbb{F}_2^n$ is a binary version of $y$ whose erasures are replaced by uniform random bits.
The goal of the decoder is to recover $x$ given the syndrome $s$ and the erasure locations.
For further details, see~\cite[Section 2]{connolly2024fast}.

\subsection{Stabilizer Formalism}
An $[[n,k]]$ quantum stabilizer code is an error correction code that protects $k$ logical qubits with $n$ physical qubits. 
The pure quantum state of a single qubit can be represented as a unit vector 
in the two-dimensional Hilbert space $\mathbb{C}_2$.
Define Pauli operators as the $2\times 2$ complex Hermitian matrices
\begin{equation}
\begin{aligned}
    &I = \begin{bmatrix}
        1 & 0 \\ 0 & 1
    \end{bmatrix},\,
    X = \begin{bmatrix}
        0 & 1 \\ 1 & 0
    \end{bmatrix},\\
    &Z = \begin{bmatrix}
        1 & 0 \\ 0 & -1
    \end{bmatrix},\,
    Y = iXZ = \begin{bmatrix}
        0 & -i\\ i & 0
    \end{bmatrix},
\end{aligned}
\end{equation}
where $i=\sqrt{-1}$. 
In an $n$-qubit system, 
given two binary vectors $a=(a_1,a_2,\ldots,a_n)\in \mathbb{F}_2^n$, and $b=(b_1,b_2,\ldots,b_n)\in \mathbb{F}_2^n$, 
we define the $n$-fold Pauli operator $D(a,b)$ as
\begin{equation}
    D(a,b) = X^{a_1}Z^{b_1}\otimes X^{a_2}Z^{b_2}\otimes \ldots \otimes X^{a_n}Z^{b_n}.
\end{equation}
Then the Pauli operators $i^k D(a,b)$ with $a,b\in\mathbb{F}_2^n$ and an overall phase $i^k$ with $k\in\{0,1,2,3\}$ form the $n$-qubit Pauli group, denoted as $\mathcal{P}_n$.

Define the \emph{symplectic inner product} between length-$2n$ binary vectors $(a,b)$ and $(a',b')$ as 
\begin{equation}
\begin{aligned}
    \langle (a,b),(a',b') \rangle_s 
    &= (a',b') \Lambda (a,b)^T \bmod 2
\end{aligned}
\end{equation}
where
\begin{equation}
    \Lambda = \begin{bmatrix} 0 & I_n \\ I_n & 0 \end{bmatrix}.
\end{equation}
This product equals 0 if the two Pauli operators $D(a,b)$ and $D(a',b')$ commute and 1 if they anti-commute.

A \emph{quantum stabilizer code} $\mathcal{C}$ with $n$ physical qubits 
is defined by an abelian subgroup $\mathcal{S}\subseteq \mathcal{P}_n$ 
called the \emph{stabilizer group} with $-I^{\otimes n}\notin \mathcal{S}$. 
The Pauli operators in $\mathcal{S}$ are referred to as the \emph{stabilizers}. 
The code space consists of all states in
$\mathbb{C}_2^{\otimes n}$ stabilized by $\mathcal{S}$:
\begin{equation}
    \mathcal{C} = \{|\psi\rangle\in\mathbb{C}_2^{\otimes n}:
    M|\psi\rangle = |\psi\rangle,\, \forall M\in\mathcal{S}\}.
\end{equation}
Code $\mathcal{C}$ has dimension $k$ if $\mathcal{S}$ has $n-k$ independent generators. 

The \emph{weight} of a Pauli operator in $\mathcal{P}_n$
is defined as the number of nontrivial elements in its $n$-fold Kronecker product that are not equal to $I$.
For example, the Pauli operator $X\otimes I\otimes Z$ in $\mathcal{P}_3$ has weight 2.
We define the \emph{distance} of a stabilizer code $\mathcal{C}$ as the minimum weight of all Pauli operators in $N(\mathcal{S})\backslash \mathcal{S}$, 
where $N(\mathcal{S})$ denotes the normalizer group of $\mathcal{S}$ in $\mathcal{P}_n$.
Code $\mathcal{C}$ is called an $[[n,k,d]]$ stabilizer code if it has distance $d$, and it is called \emph{degenerate} if $d$ is larger than the minimum weight of its stabilizers.

In the symplectic representation, the stabilizer group $\mathcal{S}$ is constructed from the rows of the stabilizer matrix 
\begin{equation}
H = [H_x,H_z],
\end{equation}
where $H_x,H_z\in\mathbb{F}_2^{m\times n}$ are binary matrices with $m$ rows and $n$ columns.
In particular, each row $(h_x,h_z)$ of $H$ defines the stabilizer $D(h_x,h_z)$ and the set of stabilizers defined by all rows generates the stabilizer group $\mathcal{S}$.

An important class of stabilizer codes, known as 
Calderbank–Shor–Steane (CSS) codes \cite{calderbank1996good,steane1996multiple} 
has each stabilizer in the form $D(a,0)$ (i.e., only $X$ operators) 
or $D(0,b)$ (i.e., only $Z$ operators).
In this case, we have
\begin{equation}
H_x = \begin{bmatrix} 0 \\ H_2 \end{bmatrix}, \quad H_z = \begin{bmatrix} H_1 \\ 0 \end{bmatrix},
\label{eq:H1_G2}
\end{equation}
Thus, the stabilizer matrix has the form
\begin{equation}
    H = \begin{bmatrix}
        0 & H_1 \\
        H_2 & 0
    \end{bmatrix}.
\end{equation}
For CSS codes, the commutativity constraint for the stabilizers given by $H \Lambda H^T = 0$ reduces to $H_2H_1^T = 0$. 
Throughout this work, we will focus our discussion on the CSS codes.

\subsection{Syndrome Decoding over Erasures}
\label{subsec:syndrome_decoding_over_erasures}
In this work, we consider the quantum erasure channel \cite{bennett1997capacities,grassl1997codes} as our noise model, 
where each qubit in the encoded state of an $[[n,k]]$ stabilizer 
code is erased independently with probability $p$.
Upon erasure, the density matrix of the qubit $\rho$ is replaced 
by the completely mixed state $I/2$. 
By writing the completely mixed state as
\begin{equation}
    I/2 = \frac{1}{4}\left(\rho + X\rho X + Y\rho Y + Z\rho Z\right),
\end{equation}
we can interpret it as the erased qubit affected by a Pauli $I$, $X$, $Y$, or $Z$ error, chosen uniformly at random.
Denote $\varepsilon\in\mathbb{F}_2^n$ as a length-$n$ \emph{erasure pattern} vector where $\varepsilon_i = 1$ if the $i$-th qubit is erased.
Then we can view the encoded state $|\psi\rangle$ as affected by 
an $n$-qubit Pauli error $E=D(x,z)\in\mathcal{P}_n$ via 
$|\psi\rangle\rightarrow E|\psi\rangle$, 
where all the non-trivial elements in $E$ are located 
in the support of the erasure pattern $\varepsilon$.
In other words, 
we have both $\mathrm{supp}(x)\subseteq\mathrm{supp}(\varepsilon)$ and 
$\mathrm{supp}(z)\subseteq\mathrm{supp}(\varepsilon)$, where $\supp(\cdot)$ 
denotes the set of nonzero locations of a vector.

Similar to the syndrome decoding for Pauli errors, 
the goal of the decoder is to correct $E$ by measuring the stabilizers in $H$.
Measuring a stabilizer $M = D(a,b) \in \mathcal{S}$ gives the symplectic inner product $\langle (a,b),(x,z)\rangle_s$, showing whether $E$ commutes or anti-commutes with $M$.
By measuring all the stabilizers in $H$, we obtain a length-$m$ binary syndrome vector 
\begin{equation}
s = (x,z) \Lambda H^T.
\end{equation}
The difference is that over the quantum erasure channel, 
the decoder has additional knowledge of the erasure pattern $\varepsilon$. 

Given a syndrome $s$ and an erasure pattern $\varepsilon$, there can be multiple Pauli errors in the support of $\varepsilon$ that matches the syndrome.
Denote the set of those Pauli errors as $\mathcal{E}(\varepsilon;s)$.
In our erasure error model, all Pauli errors in 
$\mathcal{E}(\varepsilon;s)$ 
are equally likely.
For the stabilizer code, denote the coset of its stabilizer group $\mathcal{S}$ shifted by a Pauli error $E\in\mathcal{P}_n$ as $E\mathcal{S}$. 
Since Pauli errors that differ by a stabilizer 
affects the logical state of the encoded qubits equivalently, 
the optimal decoding strategy would be finding the coset $E\mathcal{S}$ that has the largest intersection with $\mathcal{E}(\varepsilon;s)$.
However, as discussed in \cite[Sec. III]{delfosse2020linear}, 
we know the size of this intersection, 
if nontrivial, 
only depends on how the stabilizer group $\mathcal{S}$ intersects with $\mathcal{E}(\varepsilon;0)$:
\begin{equation}
    |E\mathcal{S} \cap \mathcal{E}(\varepsilon;s)| = 
    |\mathcal{S} \cap \mathcal{E}(\varepsilon;0)|.
\end{equation}
This means that all cosets of the stabilizer group have the same intersection with $\mathcal{E}(\varepsilon;s)$ and are equally likely \cite[Lemma 1]{delfosse2020linear}.
Therefore, the optimal decoding strategy is simply finding an arbitrary Pauli error $\widehat{E}$ in $\mathcal{E}(\varepsilon;s)$ as the decoding output, as the cosets represented by different choices of $\widehat{E}$ are all equally likely. 
After that, we apply it to correct the affected qubits as $E|\psi\rangle\rightarrow \widehat{E}E|\psi\rangle$. 

This decoding process has four possible outcomes:
\begin{enumerate}
    \item {\bf Failure}: 
    The decoder fails to find an $\widehat{E}$ in the support of $\varepsilon$ that matches the syndrome $s$.
    \item {\bf Successful (Exact Match)}:
    The estimated error is equal to the channel error, 
    i.e., $\widehat{E} = E$.
    \item {\bf Successful (Degenerate Error)}:
    The difference between the estimated error and the channel error 
    is a stabilizer, i.e., $\widehat{E}E\in\mathcal{S}$.
    \item {\bf Failure (Logical Error)}: 
    The difference between the estimated error and the channel error
    is a logical operator, 
    i.e., $\widehat{E}E\in N(\mathcal{S})\backslash\mathcal{S}$.
\end{enumerate}
Over the quantum erasure channel, 
it is impossible to distinguish between 
a degenerate error and a logical error since they are equally likely. 
Therefore, the optimal strategy is to minimize the likelihood 
of failure due to the first outcome, 
where the decoder fails to find an error matching the syndrome.

For CSS codes, $H_1$ in $H_z$ defines the $Z$-stabilizers that interact with Pauli-$X$ errors, and $H_2$ in $H_x$ defines the $X$-stabilizers that interact with Pauli-$Z$ errors.
Thus for encoded state affected by a Pauli error $E=D(x,z)$, 
we can separate the syndrome vector into two parts $s = (s_x,s_z)$ where $s_x = x H_1^T$ and $s_z = z H_2^T$. 
In this way, the decoding problem can be divided into correcting the Pauli-$X$ error 
$D(x,0)$ and the Pauli-$Z$ error $D(0,z)$ separately, 
both with knowledge of the erasure pattern $\varepsilon$.
In this work, we focus on correcting Pauli-$X$ errors using measurements of $Z$-stabilizers, as the correction of Pauli-$Z$ errors follows the exact same process.

\subsection{Gaussian Elimination and Peeling}
Consider decoding an $[[n,k]]$ CSS code corrupted by a Pauli-$X$ error $D(x,0)$ in the support of an erasure pattern $\varepsilon$. 
This is essentially the same as solving the implied erasure syndrome decoding problem for a classical linear code with parity check matrix $H_1$. 
As discussed in Section~\ref{subsec:syndrome_decoding_over_erasures}, an ML decoder can output any $\widehat{x}$ that matches the syndrome. 
This can be achieved by solving a system of linear equations 
$\widehat{x}H_1^T = s_x$ with 
$\supp(\widehat{x})\subseteq\supp(\varepsilon)$ over $\mathbb{F}_2$
using Gaussian Elimination. 
We refer to it as the \emph{Gaussian Elimination decoder}, 
which has a time complexity of $O(\nu^3)$ 
if there are $\nu$ erased qubits.
The constants in its complexity can be reduced
by leveraging the sparsity of $H_1$ \cite{burshtein2004efficient}. 
However, this complexity is generally considered high for codes with moderate block lengths \cite{bocharova2018ldpc}.
Therefore, we first turn to the classical {peeling decoder}  \cite{luby2001efficient,richardson2001efficient} with linear complexity 
that runs on a Tanner graph.

A Tanner graph $\mathcal{G} = (\mathcal{V},\mathcal{C},\mathcal{E})$ 
is a bipartite graph with the variable nodes in 
$\mathcal{V}=\{v_1,\ldots,v_n\}$ representing the elements 
of $x=(x_1,\ldots,x_n)$ and the check nodes in 
$\mathcal{C}=\{c_1,\ldots,c_m\}$ representing 
the $Z$-stabilizers in $H_1$.
A variable node $v_i$ is connected to a check node $c_j$ 
if $H_1 (i,j)=1$.

For a Tanner graph $\mathcal{G}$, we define:
\begin{itemize}
    \item {\bf dangling check}: a check node with degree one.
    \item {\bf variable-induced subgraph}: a subgraph of $\mathcal{G}$ 
    induced by a set of variable nodes and their connected check nodes.
    \item {\bf stopping set}: a set of variable nodes 
    whose variable-induced subgraph has no dangling checks.
\end{itemize}
An example Tanner graph with a variable-induced subgraph 
of a stopping set is shown in Fig.~\ref{fig:tanner}.

\usetikzlibrary{shapes.geometric}
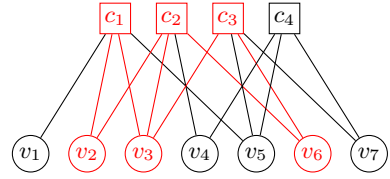
\begin{figure}[t]
    \centering
    \begin{tikzpicture}[scale=1.5,
    square/.style={regular polygon, regular polygon sides=4}]
    \tikzset{
        v_node/.style={draw,circle,inner sep=0.5mm, font=\small},
        c_node/.style={draw,square,inner sep=0.0mm,font=\small},
    }
    \foreach \x in {1,4,5,7}
        \node[v_node] (v\x) at (\x/2,0) {$v_\x$};
    \foreach \x in {2,3,6}
        \node[v_node, color=red] (v\x) at (\x/2,0) {$v_\x$};
    \foreach \x in {1,2,3}
        \node[c_node,color=red] (c\x) at (\x/2+1.5/2,1.2) {$c_\x$};
    \foreach \x in {4}
        \node[c_node] (c\x) at (\x/2+1.5/2,1.2) {$c_\x$};

    \draw (c1)--(v1);
    \draw[color=red] (c1)--(v2);
    \draw[color=red] (c1)--(v3);
    \draw (c1)--(v5);
    
    \draw[color=red] (c2)--(v2);
    \draw[color=red] (c2)--(v3);
    \draw (c2)--(v4);
    \draw[color=red] (c2)--(v6);
    
    \draw[color=red] (c3)--(v3);
    \draw (c3)--(v5);
    \draw[color=red] (c3)--(v6);
    \draw (c3)--(v7);
    
    \draw (c4)--(v4);
    \draw (c4)--(v5);
    \draw (c4)--(v7);
    \end{tikzpicture}
    \caption{Example of a Tanner graph and a variable-induced subgraph of a stopping set $\{v_2,v_3,v_6\}$.}
    \label{fig:tanner}
\end{figure}

To initialize the peeling decoder, 
we label the check nodes with their syndrome bits and consider the 
variable-induced subgraph $\mathcal{G}'$ for the erased qubits.
The peeling decoder works by iteratively resolving error bits connecting to dangling checks, updating the residual syndrome, 
and then removing the resolved variable node and 
its adjacent edges from $\mathcal{G}'$.
The pseudo-code for the peeling decoder is provided in 
Algorithm~\ref{alg:peeling}, 
in which the set of check nodes connected to a variable node 
$v_j$ is denoted by $\partial v_j$. 

\begin{algorithm}[t]
\caption{Peeling decoding}
\label{alg:peeling}
\KwIn{
block length $n$, 
Tanner graph $\mathcal{G}=(\mathcal{V},\mathcal{C},\mathcal{E})$ 
for $H_1$, syndrome $s_x$, erasure pattern $\varepsilon$}
\KwOut{$\widehat{x}$, residual syndrome $s'_x$, 
residual Tanner graph $\mathcal{G}'$}
$\widehat{x}\leftarrow$ all-zero vector of length $n$\\
$s_x' \leftarrow s_x$\\
$\mathcal{G}'\leftarrow$ variable-induced subgraph 
of $\varepsilon$ in $\mathcal{G}$\\ 
$\mathcal{D} \leftarrow\{\text{dangling checks in }\mathcal{G}'\}$\\
\While{$\mathcal{D}$ is not empty}{
    $c_i\leftarrow$ a dangling check in $\mathcal{D}$\\
    $v_j\leftarrow$ variable node connected by $c_i$ in $\mathcal{G}'$\\
    $\widehat{x}_j\leftarrow (s'_x)_j$\\
    $(s'_{x})_k\leftarrow (s'_{x})_k + \widehat{x}_j$ 
    for all $c_k\in \partial x_j$\\
    Check all $c_k\in \partial x_j$ in $\mathcal{G}'$ to 
    update $\mathcal{D}$\\
    Remove $c_i,v_j$ and their adjacent edges from $\mathcal{G}'$
}
\Return{$\widehat{x}$, $s'_x$, $\mathcal{G}'$}
\end{algorithm}

This peeling process continues until either all erasures are recovered, 
or there are no more dangling checks in 
$\mathcal{G}'$. 
In the former, the erasure pattern $\varepsilon$ is called \emph{peelable}. 
In that case,
the residual Tanner graph $\mathcal{G}'$ will be empty and 
the residual syndrome $s'_x$ will be all-zero.
Otherwise, the residual Tanner graph $\mathcal{G}'$ will be the 
variable-induced subgraph of the maximal {stopping set} 
covered by the erasures \cite{di2002finite}.

\section{Erasure Decoding via Cluster Decomposition}
\label{sec:cluster}

The peeling decoder performs well for classical LDPC codes over erasures \cite{luby2001efficient}.
However, directly applying the peeling decoder to quantum LDPC codes often results in degraded performance due to code degeneracy. 
In the case of CSS quantum LDPC codes, $X$-stabilizers represented by the rows in the matrix $H_2$ generally have low weights. 
Consequently, when decoding Pauli-$X$ errors on randomly erased qubits, it is common for the erasure pattern $\varepsilon$ to cover one or more $X$-stabilizers. 
In these cases, there are multiple errors within the support of $\varepsilon$ that match the syndrome, each differing from the others by an $X$-stabilizer. 
From an ML decoding perspective, this is not an issue, as errors differing by stabilizers lead to the same decoding outcome. 
However, the qubits associated with these $X$-stabilizers form stopping sets in the Tanner graph of $H_1$ \cite[Lemma 1]{connolly2024fast}, which obstructs the peeling decoding process.
A straightforward approach to improving the peeling decoder is to augment it with Gaussian Elimination to resolve the remaining erasures in the stopping sets. However, directly applying Gaussian Elimination as a post-processing step requires cubic time complexity in the size of the stopping set.

For hypergraph product codes, Connolly \emph{et al.} observed that after pruned peeling, which is peeling combined with a stabilizer search, most remaining erasures are due to stopping sets in the classical component codes \cite{connolly2024fast}. 
Based on this observation, they proposed decomposing the stopping sets into vertical and horizontal clusters by constructing a VH graph, and solving those clusters sequentially. 
This decomposition transforms the problem of solving the remaining erasures into multiple smaller subproblems, each corresponding to a vertical or horizontal stopping set of the classical component code, thereby reducing overall complexity. 
One future direction they mentioned is extending VH decomposition to other families of codes, which seems challenging, as VH specifically exploits the product structure of hypergraph product codes. 

In this section, we introduce a post-processing procedure for peeling, called \emph{cluster decomposition}, which generalizes the stopping set decomposition approach of VH and applies to arbitrary quantum LDPC codes. 
When the peeling decoder fails to fully recover the erasures, we compute the biconnectivity structure of the residual Tanner graph $\mathcal{G}'$ and decompose it into a cluster forest. 
This decomposition reduces the problem of solving the remaining erasures to several smaller subproblems that can be solved sequentially.

We refer to the combination of peeling and cluster decomposition as the cluster decoder. By allowing clusters to have unconstrained sizes, the cluster decoder achieves maximum likelihood (ML) decoding performance with reduced complexity compared to direct Gaussian Elimination. When the cluster sizes are constrained to a constant, we show that the cluster decoder achieves linear decoding complexity.

The cluster decoder operates in three phases: 
\begin{enumerate} 
    \item \textbf{Peeling on the Tanner graph}: 
    Attempt to recover the erasure pattern using a peeling decoder. If successful, terminate the process. 
    \item \textbf{Cluster decomposition}: 
    If the erasure pattern is not peelable, decompose the residual Tanner graph into a cluster forest. 
    \item \textbf{Constructing cluster solutions}: 
    Compute solutions for the clusters in each cluster tree sequentially using Gaussian Elimination, then merge these solutions into a global solution for the remaining erasures. 
\end{enumerate}

We note that in the first phase of our algorithm, peeling can be replaced with pruned peeling proposed in \cite{connolly2024fast} for potential improvements. 
The pruned peeling approach is based on the idea that, if an $X$-stabilizer is entirely covered by the erasure pattern, we are free to fix one bit in its support to 0. 
When peeling reaches a stopping set, pruned peeling suggests searching for linear combinations of up to $M$ stabilizers which are rows in the $X$-stabilizer matrix $H_2$ that are fully covered by erasures and fixing one of their bits to 0. 
This operation can potentially resolve the stopping set or reactivate the peeling procedure, with an additional linear time complexity.

In this work, however, we only focus on peeling combined with cluster decomposition. 
Our intuition is that in the low-erasure-rate regime, cluster decomposition can handle cases involving low-weight $X$-stabilizers within stopping sets by identifying them as small clusters. 
We leave a more comprehensive study on pruned peeling combined with cluster decomposition for future work.

In the remainder of this section, we provide a detailed description of the second and third phases of the algorithm.

\subsection{Cluster Decomposition}
Consider decoding a Pauli-$X$ error $D(x,0)$ for 
a CSS codes with $\supp(x)\subseteq\supp(\varepsilon)$ 
given an erasure pattern $\varepsilon$.
If $\varepsilon$ is not peelable, 
after applying the peeling decoder in Algorithm~\ref{alg:peeling},
we are left with a non-zero residual syndrome $s'_x$
and a residual Tanner graph $\mathcal{G}'$, which is 
a variable-induced subgraph of the unresolved erasures.
Denote the residual erasure pattern as $\varepsilon'$,
the remaining problem is to find a solution $\widehat{x}'$ in the support 
of $\varepsilon'$ that match the residual syndrome $s'_x$ as 
$\widehat{x}'H_1^T = s'_x$.

We start by decomposing the residual Tanner graph $\mathcal{G}'$ into a cluster forest. 
As outlined in Section \ref{subsec:graphs}, 
this can be achieved by first applying the Hopcroft-Tarjan's algorithm on $\mathcal{G}'$ 
to identify its set of clusters and cut nodes, 
and then connecting clusters to cut nodes 
whenever a cut node is contained within a cluster.
Then, we can suspend the tree from an arbitrary cluster $B_r$, 
and view it as a cluster tree {rooted} at $B_r$.
After assigning one cluster as the root, 
the parent-child relation between the clusters and cut nodes 
in the tree will be uniquely determined.
By construction, the cluster tree alternates cut node levels 
to cluster levels and always has clusters as leaves.
An example of the cluster decomposition 
of a stopping set is provided in Figure~\ref{fig:cluster_tree}.
Note that a cluster is not necessarily a variable-induced subgraph. 
For example, cluster $B_r$ 
in Figure~\ref{fig:cluster_tree} is not a variable-induced subgraph 
of $\{v_4,v_5,v_6\}$.

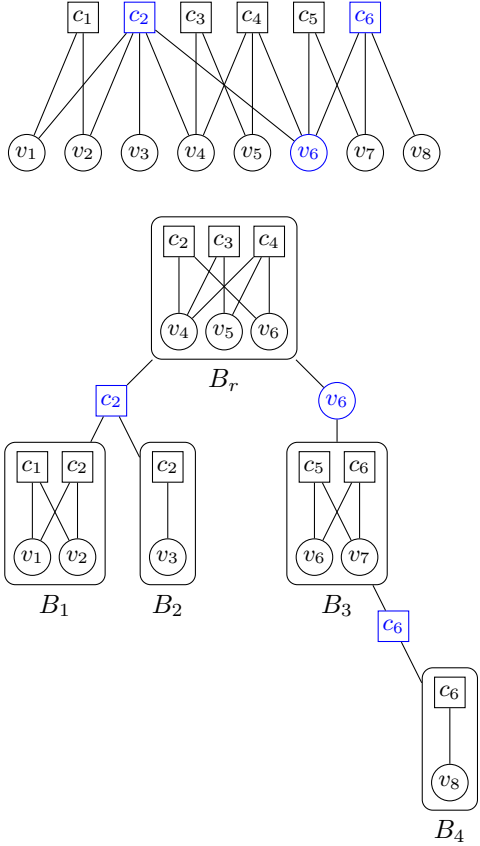
\begin{figure}[t]
    \centering
    \begin{tikzpicture}[scale=1.5,
    square/.style={regular polygon, regular polygon sides=4},]
    \tikzset{
        v_node/.style={draw,circle,inner sep=0.5mm, font=\small},
        c_node/.style={draw,square,rounded corners=false,solid,inner sep=0.0mm,font=\small},            
        blocknode/.style={draw,rectangle,rounded corners,font=\small},
        cutnode/.style={draw,circle,font=\small},
    }                                                                       
    \foreach \x in {1,2,3,4,5,7,8}                                            
        \node[v_node] (v\x) at (\x/2,0) {$v_\x$};
    \foreach \x in {6}                                            
        \node[v_node, color=blue] (v\x) at (\x/2,0) {$v_\x$};
    \foreach \x in {1,3,4,5}
        \node[c_node] (c\x) at (\x/2+1.0/2,1.2) {$c_\x$};
    \foreach \x in {2,6}
        \node[c_node, color=blue] (c\x) at (\x/2+1.0/2,1.2) {$c_\x$};

    \draw (c1)--(v1);
    \draw (c1)--(v2);
    \draw (c2)--(v1);
    \draw (c2)--(v2);
    
    \draw (c2)--(v3);
    
    \draw (c2)--(v4);
    \draw (c2)--(v6);
    \draw (c3)--(v4);
    \draw (c3)--(v5);
    \draw (c4)--(v4);
    \draw (c4)--(v5);
    \draw (c4)--(v6);
    
    \draw (c5)--(v6);
    \draw (c5)--(v7);
    \draw (c6)--(v6);
    \draw (c6)--(v7);
    
    \draw (c6)--(v8);

    \begin{scope}[shift={(2.25,-1.2)}]
    \node[blocknode,label=below:{$B_r$}] (R) at (0, 0) {
    \begin{tikzpicture}[scale=1.2]
        \foreach \x in {4,5,6}
            \node[v_node] (v\x) at (\x/2-3/2,0) {$v_\x$};
        \foreach \x in {2,3,4}
            \node[c_node] (c\x) at (\x/2-1/2,1.0) {$c_\x$};
        \draw (c2)--(v4);
        \draw (c2)--(v6);
        \draw (c3)--(v4);
        \draw (c3)--(v5);
        \draw (c4)--(v4);
        \draw (c4)--(v5);
        \draw (c4)--(v6);
    \end{tikzpicture}
    };
    
    \node[blocknode,label=below:{$B_1$}] (A) at (-1.5,-2) {
    \begin{tikzpicture}[scale=1.2]
        \foreach \x in {1,2}
            \node[v_node] (v\x) at (\x/2-1/2,0) {$v_\x$};
        \foreach \x in {1,2}
            \node[c_node] (c\x) at (\x/2-1/2,1.0) {$c_\x$};
            \draw (c1)--(v1);
            \draw (c1)--(v2);
            \draw (c2)--(v1);
            \draw (c2)--(v2);
    \end{tikzpicture}
    };
    \node[blocknode,label=below:{$B_2$}] (B) at (-0.5,-2) {
    \begin{tikzpicture}[scale=1.2]
        \foreach \x in {3}
            \node[v_node] (v\x) at (\x/2-3/2,0) {$v_\x$};
        \foreach \x in {2}
            \node[c_node] (c\x) at (\x/2-2/2,1.0) {$c_\x$};
            \draw (c2)--(v3);
    \end{tikzpicture}
    };
    \node[blocknode,label=below:{$B_3$}] (C) at (1,-2) {
    \begin{tikzpicture}[scale=1.2]
        \foreach \x in {6,7}
            \node[v_node] (v\x) at (\x/2-6/2,0) {$v_\x$};
        \foreach \x in {5,6}
            \node[c_node] (c\x) at (\x/2-5/2,1.0) {$c_\x$};
        \draw (c5)--(v6);
        \draw (c5)--(v7);
        \draw (c6)--(v6);
        \draw (c6)--(v7);
    \end{tikzpicture}
    };
    \node[blocknode,label=below:{$B_4$}] (D) at (2,-4) {
    \begin{tikzpicture}[scale=1.2]
        \foreach \x in {8}
            \node[v_node] (v\x) at (\x/2-8/2,0) {$v_\x$};
        \foreach \x in {6}
            \node[c_node] (c\x) at (\x/2-6/2,1.0) {$c_\x$};
        \draw (c6)--(v8);
    \end{tikzpicture}
    };
    
    \node[c_node, color=blue] (i) at (-1,-1.0) {$c_2$};
    \node[v_node, color=blue] (j) at (1,-1.0) {$v_6$};
    \node[c_node, color=blue] (k) at (1.5,-3.0) {$c_6$};
    
    \draw (R) -- (i) -- (A);
    \draw (i) -- (B);
    \draw (R) -- (j) -- (C) -- (k) -- (D);
    \end{scope}
    
    \end{tikzpicture}
    \caption{Example of the cluster decomposition 
    of the variable-induced subgraph of a stopping set. 
    On top, we show the subgraph $\mathcal{G}'$ of a stopping set 
    $\{v_1,\ldots,v_8\}$ that has no dangling checks. 
    Graph $\mathcal{G}'$ can be decomposed into five clusters 
    $B_r=\{c_2,c_3,c_4,v_4,v_5,v_6\}$, 
    $B_1=\{c_1,c_2,v_1,v_2\}$, 
    $B_2=\{c_2,v_3\}$, 
    $B_3=\{c_5,c_6,v_6,v_7\}$, and
    $B_4=\{c_6,v_8\}$.
    They are connected by three cut nodes $c_2$, $v_6$, and $c_6$, 
    highlighted in blue.
    At the bottom, we show the cluster tree of $\mathcal{G}'$ 
    rooted at its largest cluster.}
    \label{fig:cluster_tree}
\end{figure}

Following Lemma~\ref{lm:cluster_construction_complexity}, the construction of the cluster forest can be achieved with a time complexity of at most $O(2|\mathcal{V}|+2|\mathcal{C}|+|\mathcal{E}|)$ for a subgraph $\mathcal{G}'$ of the Tanner graph $\mathcal{G} = (\mathcal{V},\mathcal{C},\mathcal{E})$. 
For a CSS quantum LDPC code where the weights of $Z$-stabilizers 
in $H_1$ are bounded by a constant, 
the degrees of the check nodes in the Tanner graph $\mathcal{G}$
are bounded by the same constant. 
In such cases, the total number of edges $|\mathcal{E}|$ in $\mathcal{G}$ 
is upper bounded by the maximum check degree multiplied by 
the number of check nodes, making $|\mathcal{E}|$ linear in the block length.
Therefore, for a CSS quantum LDPC code, 
the time complexity of constructing the cluster forest, 
as the first step in our post-processing procedure, 
is linear in the block length of the code.

After the cluster decomposition, the problem of finding a solution 
in the residual erasure pattern can be divided into several 
subproblems of finding solutions for each of the clusters.
Next, we show how to construct solutions for the clusters sequentially, 
and then merge them into a single global solution on the cluster tree 
that matches the residual syndrome. 
It suffices to describe this procedure on a single cluster tree. 
In the more general case of a cluster forest, 
the same procedure can be applied to solve 
each of the isolated cluster trees independently.

\subsection{Consistency of Cluster Solutions}

For a cluster $B$ in the tree, let $\mathcal{V}(B)$ and 
$\mathcal{C}(B)$ denote the sets of variable nodes 
and check nodes contained within the cluster, respectively.
Check nodes in $\mathcal{C}(B)$ that are not cut nodes 
are referred to as \emph{internal checks}.
An assignment on the erased variable nodes within $B$, 
represented by a vector $\widehat{x}_{B}$ supported on 
$\mathcal{V}(B)$, is called a \emph{solution} of $B$ 
if it matches the residual syndrome $s'_x$ on the internal checks.
If $B$ contains a check node $c_j$ that is a cut node, 
we call the $j$-th syndrome bit of $\widehat{x}_{B}$ 
as its \emph{contribution} to the syndrome value on $c_j$.

For example, considering the cluster $B_3$ in 
Figure~\ref{fig:cluster_tree},
we have $\mathcal{V}(B_3) = \{v_6,v_7\}$ and 
$\mathcal{C}(B_3) = \{c_5,c_6\}$.
In cluster $B_3$, $c_5$ is an internal check and $c_6$ is a cut node.
A solution of $B_3$ will be a binary vector $\widehat{x}_{B_3}$ 
supported on the locations of $\{v_5,v_6\}$ 
that match the residual syndrome on $c_5$.
The contribution of $\widehat{x}_{B_3}$ 
to the syndrome value on $c_6$ equals
\begin{equation}
(\widehat{x}_{B_3})_5+(\widehat{x}_{B_3})_6\bmod 2.
\end{equation}

Given a cluster tree with a residual syndrome $s'_x$, 
we now describe the consistency conditions over cut nodes that a collection of cluster solutions must satisfy to be merged into a global solution. 
From now on, we refer to cut nodes that are variable nodes in the Tanner graph as \emph{variable cut nodes}, and cut nodes that are check nodes as \emph{check cut nodes}.

For a variable cut node $v_i$, let $v_i(B)$ denote its value in the solution of its adjacent cluster $B$. 
A collection of cluster solutions on the tree is said to be \emph{consistent} on $v_i$ if, for all its adjacent clusters 
$B_1, B_2, \ldots, B_s$, we have
\begin{equation}
v_i(B_1) = v_i(B_2) = \ldots = v_i(B_s).
\label{eq:v_cut_node_constraint}
\end{equation}
In other words, the values of $v_i$ in the solutions of all its neighboring clusters must be equal.

For a check cut node $c_j$, let $c_j(B)$ denote the contribution from the solution of its adjacent cluster $B$ to the syndrome value on $c_j$. 
A collection of cluster solutions is \emph{consistent} on $c_j$ if, for all its adjacent clusters $B_1, B_2, \ldots, B_s$, we have
\begin{equation}
(s'_x)_j = c_j(B_1) + c_j(B_2) + \ldots + c_j(B_s) \bmod 2.
\label{eq:c_cut_node_constraint}
\end{equation}
In other words, the sum of contributions from all its neighboring clusters to the syndrome value on $c_j$ must equal the residual syndrome bit $(s'_x)_j$ on $c_j$.

A collection of cluster solutions on the tree is called \emph{globally consistent} if it is \emph{consistent} on every cut node. 
In this case, the cluster solutions can be merged into a global solution that matches the residual syndrome $s'_x$.

\subsection{Constructing Cluster Solutions}
We now show how to construct cluster solutions 
that are globally consistent. 
The algorithm consists of two stages. 
In the first stage, we compute solutions for each cluster, 
starting from the leaves and progressing up to the root, 
while propagating constraints along the edges of the tree. 
For a cluster $B$ with $n'$ variable nodes and $m'$ check nodes, 
one of its solutions can be computed in time $O(n'm'\min(n',m'))$.
with Gaussian Elimination. 
During this process, each cluster collects constraints 
passed from its child cut nodes if any, computes its solutions, 
and summarizes the constraint it needs 
to propagate upwards to its parent cut node.  
Similarly, each cut node collects and summarizes constraints 
from its child clusters and passes the resulting constraint 
to its parent cluster. 
Here we describe this process in detail.

Consider a non-root cluster $B$ with a variable cut node $v_i$ as its parent. 
After collecting constraints from its child cut nodes if any,
we compute two potential solutions for $B$, which are stored for later use: 
one with $v_i(B) = 0$ and another with $v_i(B) = 1$. 
If both solutions exist, $B$ is labeled \emph{free} for its parent $v_i$, 
meaning that regardless of the value assigned to $v_i$, 
there is a corresponding solution in $B$. 
If only one solution exists, for example, when $v_i(B) = 0$, 
then $B$ is labeled \emph{frozen} and it propagates the constraint 
$v_i = 0$ upwards to $v_i$. 
We adopt the terms \emph{free} and \emph{frozen} from 
\cite{connolly2024fast} to describe the type of 
constraints clusters propagate to their parent cut nodes.

Similarly, consider a non-root cluster $B$ with a check cut node $c_j$ as its parent. 
After collecting constraints from its child cut nodes if any, 
we compute two potential solutions for $B$, which are stored for later use: 
one with $c_j(B) = 0$ and another with $c_j(B) = 1$. 
If both solutions exist, $B$ is labeled \emph{free} for $c_j$, 
meaning that, regardless of the contribution required for the syndrome on $c_j$, 
a corresponding solution exists in $B$. 
If only one solution exists, for instance, when $c_j(B) = 0$, 
$B$ is labeled \emph{frozen} and its only possible contribution, 
$c_j(B) = 0$, is propagated upwards to $c_j$.

Next, we describe how constraints are collected and summarized by the cut nodes. 
For a variable cut node $v_i$, after solutions for all its child clusters 
have been computed, it gathers their constraints as follows. 
If $v_i$ has a frozen child cluster that sends a constraint, 
such as $v_i = 0$, then following (\ref{eq:v_cut_node_constraint}),
$v_i$ is labeled \emph{frozen} for its parent cluster $B$ 
and propagates the same constraint upwards. 
If all child clusters are labeled \emph{free}, 
$v_i$ is also labeled \emph{free} for its parent $B$.
We note that if $v_i$ receives more than one constraint 
from its frozen child clusters, they must be consistent.
Since otherwise it means that the space of solutions consistent on $v_i$ is empty, while we know for the erasure syndrome decoding problem, 
at least one globally consistent solution exists.

For a check cut node $c_j$, if any of its child clusters are labeled \emph{free} 
regarding their possible contributions to the syndrome on $c_j$, 
then the required contribution from its parent cluster $B$ is also labeled \emph{free}.
Otherwise, when all its child clusters $B_{1},\ldots,B_{s}$ have \emph{frozen} 
contributions, to match the residual syndrome on $c_j$, 
the contribution required from its parent $B$ can be computed from 
(\ref{eq:c_cut_node_constraint}) as
\begin{equation}
c_j(B) = (s'_x)_i - c_j(B_{1}) - \cdots - c_j(B_{s}) \bmod 2. 
\end{equation}
In this case, $c_j$ is labeled \emph{frozen} and propagates
its required contribution to its parent $B$.

Now we describe how a cluster $B$ collects constraints 
from its child cut nodes before computing its solutions.
There are four possible cases:
\begin{enumerate}
    \item If $B$ has a frozen variable cut node $v_i$ as a child, 
    the received constraint, such as $v_i(B)=0$, 
    is incorporated when computing the solutions for $B$.
    \item If $B$ has a free variable cut node $v_i$ as a child, 
    no additional constraint is applied.
    \item If $B$ has a frozen check cut node $c_j$ as a child, 
    the required contribution, such as $c_j(B) = 0$, 
    is incorporated when computing the solutions for $B$.
    \item If $B$ has a free check cut node $c_j$ as a child, 
    the contribution from $B$ to $c_j$ can be arbitrary, 
    and no additional constraint is applied.
\end{enumerate}

A recursive implementation that collects the incoming constraints from its children and computes the solutions for a cluster $B$ is provided in Algorithm~\ref{alg:recursive_compute}.

\begin{algorithm}[t]
\caption{RecursiveCompute}
\label{alg:recursive_compute}
\KwIn{Cluster $B$, residual syndrome $s'_x$}
$L\leftarrow$ list of additional constraints\\
\tcp{Collect constraints, if any}
\For{$v_i$ in the child variable cut nodes of $B$}{
    \For{$B'$ in the child clusters of $v_i$}{
        RecursiveCompute($B'$, $s'_x$)\\
    }
    \If{exist a frozen child cluster $B'$}{
        add $v_i(B)=v_i(B')$ to $L$
    }
}
\For{$c_j$ in the child check cut nodes of $B$}{
    \For{$B'$ in the child clusters of $c_j$}{
        RecursiveCompute($B'$, $s'_x$)\\
    }
    \If{all child clusters $B_1,\ldots,B_s$ are frozen}{
        add $c_j(B)=(s_x')_j-c_j(B_1)-\ldots-c_j(B_s)\bmod 2$ to $L$
    }
}
\tcp{Compute cluster solutions}
\eIf{$B$ is not a root cluster}{
    $n\leftarrow$ parent cut node of $B$\\
    Compute and store two solutions of $B$ with $L$, one with $n(B)=0$, and the other with $n(B) = 1$\\
    \eIf{both solutions exist}{
        $B$ labeled \emph{free} for parent $n$
    }{
        $B$ labeled \emph{frozen} for parent $n$
    }
}{
    Compute one solution of $B$ with $L$
}
\end{algorithm}

After this process reaches the root cluster $B_r$, 
we compute a single solution for $B_r$ incorporating 
all received constraints from its child cut nodes. 
At this point, every frozen cluster has one solution stored, 
and every free cluster has two solutions stored.
In the second stage of our algorithm, we select a solution 
for each cluster by propagating constraints backward from root to leaves.

For a cluster $B$ in the tree, if it is the root cluster 
or a \emph{frozen} cluster, there is only one solution to select. 
If $B$ is a free cluster with two stored solutions, 
we choose the solution consistent with the constraint 
propagated downward from its parent cut node. Specifically,
if $B$ has a variable cut node $v_i$ as its parent, 
we select the solution whose value on $v_i$
is consistent with the received constraint;
if $B$ has a check cut node $c_j$ as its parent, 
we select the solution consistent with the required contribution. 
Afterward, for each child variable cut node $v_k$, 
$B$ propagates the constraint $v_k = v_k(B)$ 
of its selected solution downward to $v_k$. 
Similarly, for each child check cut node $c_k$, 
$B$ propagates the contribution $c_k(B)$ of its solution downward to $c_k$.

Next, we describe how constraints are propagated and resolved 
through the cut nodes. 
For a variable cut node $v_i$, after receiving the constraint, 
such as $v_i=0$, from its parent cluster $B$, 
it simply propagates the same constraints downwards 
for each of its child clusters.
The situation for check cut nodes is a bit more subtle.
For a check cut node $c_j$, 
let $B_1,\ldots,B_t$ denote its child clusters with free contributions, 
and let $B^\ast_1,\dots,B^\ast_s$ denotes its child clusters with frozen contributions.
After receiving the contribution $c_j(B)$ from its parent cluster $B$, 
we need to assign a required contribution 
for each of its free clusters to match the residual syndrome on $c_j$.
First, we compute the sum of the required contributions
from all its free child clusters, denoted by $\Delta s$, 
following (\ref{eq:c_cut_node_constraint}) as
\begin{equation}
    \Delta s_j = (s'_x)_j - c_j(B) - c_j(B^\ast_1) - \ldots - c_j(B^\ast_s)
    \bmod 2
\end{equation}
Then, to match the residual syndrome on $c_j$, 
we need a contribution assignment with
\begin{equation}
    \Delta s_j = c_j(B_1) + \ldots + c_j(B_s) \bmod 2,
\end{equation}
There may be more than one assignment that satisfies the requirement. 
In our implementation, we simply assign the required contribution to the first free cluster $B_1$ to be $c_j(B_1) = \Delta s_j$, with zero contribution assigned to the rest of the free clusters.
A recursive implementation that selects a solution for a cluster $B$ and propagates its constraints downwards to its children is provided in Algorithm~\ref{alg:recursive_select}.

\begin{algorithm}[t]
\caption{RecursiveSelect}
\label{alg:recursive_select}
\KwIn{Cluster $B$ with parent cut node $n$, constraint $n(B) = b$, residual syndrome $s'_x$}
\tcp{Select a cluster solution}
\eIf{$B$ is frozen or it is the root cluster}{
    $B$ has only one solution to select from
}{
    $B$ has two stored solutions, select the one consistent with 
    $n(B) = b$
}
\tcp{Propagate constraints}
\For{$v_i$ in the child variable cut nodes of $B$}{
    \For{$B'$ in the child clusters of $v_i$}{
        RecursiveSelect($B'$, $v_i(B') = v_i(B)$, $s'_x$)\\
    }
}
\For{$c_j$ in the child check cut nodes of $B$}{
    $B^\ast_1,\ldots,B^\ast_s\leftarrow$ frozen child clusters of $c_j$\\
    $B_1,\ldots,B_t\leftarrow$ free child clusters of $c_j$\\
    Compute $\Delta s_j = (s'_x)_j - c_j(B) - c_j(B^\ast_1) - \ldots - c_j(B^\ast_s)\bmod 2$\\
    \tcp{Frozen child clusters}
    \For{$B'$ in \{$B^\ast_2,\ldots,B^\ast_s$\}}{
        RecursiveSelect($B'$, $c_j(B')$ is frozen, $s'_x$)
    }
    \tcp{Free child clusters}
    RecursiveSelect($B_1$, $c_j(B_1) = \Delta s_j$, $s'_x$)\\
    \For{$B'$ in \{$B_2,\ldots,B_t$\}}{
        RecursiveSelect($B'$, $c_j(B_1) = 0$, $s'_x$)
    }
}
\end{algorithm}

Once this process in the second stage reaches the leaves of the cluster tree, a globally consistent solution will have been selected for each cluster. 
The cluster solutions are then merged into a global solution $\widehat{x}'$
that satisfies the residual syndrome, such that $\widehat{x}'H_1^T = s'_x$. 
The pseudocode summarizing the entire cluster decomposition post-processing procedure  
is provided in Algorithm~\ref{alg:cluster_decomp}. 

\begin{algorithm}[t]
\caption{Cluster decomposition post-processing}
\label{alg:cluster_decomp}
\KwIn{
block length $n$, 
residual Tanner graph $\mathcal{G}'$, 
residual syndrome $s'_x$}
\KwOut{$\widehat{x}'$}
\tcp{Cluster decomposition}
Decompose $\mathcal{G}'$ into a cluster forest with Hopcroft-Tarjan's algorithm\\
\tcp{Construct cluster solutions}
\For{each cluster tree in the forest}{
    $B_r\leftarrow$ root cluster\\
    RecursiveCompute($B_r$, $s'_x$)\\
    RecursiveSelect($B_r$, $s'_x$)
}
Merge cluster solutions into a global solution $\widehat{x}'$\\
\Return{$\widehat{x}'$}
\end{algorithm}

\subsection{Complexity of Cluster Decomposition Post-processing}

With cluster decomposition, we effectively achieve ML decoding by resolving all remaining erasures after peeling, with reduced complexity compared to directly applying Gaussian Elimination to the entire stopping set. 
Now, we show that the complexity of cluster decomposition postprocessing, as described in Algorithm 4, is dominated by the sum of the cubics of cluster size.

\begin{theorem}
Assume the Tanner graph $\mathcal{G} = (\mathcal{V},\mathcal{C},\mathcal{E})$ has a number of edges $|\mathcal{E}|$ that is linear in $n$.
Let $B_1,B_2,\ldots,B_t$ denote all the clusters obtained after the cluster decomposition of the 
residual Tanner graph $\mathcal{G}'$. For each cluster $B_i$, let $m_i$ and $n_i$ denote the number of check nodes and variable nodes contained in $B_i$ for $i = 1,2,\ldots,t$. Then, the time complexity of Algorithm~\ref{alg:cluster_decomp} is bounded by 
\begin{equation}
    O\left(\sum_{i=1}^t n_im_i\min\{n_i,m_i\}\right) + O(n)
\end{equation}
\label{thm:complexity}
\end{theorem}

\begin{proof}
Given the Tanner graph $\mathcal{G}$ with the total number of edges $|\mathcal{E}|$ linear in the block length $n$, 
the first step in Algorithm~\ref{alg:cluster_decomp}, which constructs the cluster forest, 
takes a time complexity of $O(n)$ with Hopcroft-Tarjan's algorithm following Lemma~\ref{lm:cluster_construction_complexity}.

Then, the algorithm recursively computes the cluster solutions from the leaves to the root of each cluster tree. 
This procedure involves two calls to compute the solution for each cluster $B_i$ using Gaussian Elimination, with a time complexity of $O(n_i m_i \min\{n_i, m_i\})$ for each cluster $B_i$.
In addition, the algorithm performs a linear number of operations on the cluster forest, propagating the constraint information upwards from the leaves to the root. 

After that, the constraint information is recursively propagated downwards from the root to the leaves to select solutions for all clusters that are globally consistent. 
This step also requires a linear number of operations on the cluster forest.

Therefore, the overall time complexity of the cluster decomposition post-processing given in Algorithm~\ref{alg:cluster_decomp} is
\begin{equation}
    O\left(\sum_{i=1}^t n_im_i\min\{n_i,m_i\}\right) + O(n)
\end{equation}
\end{proof}

Define the number of variable nodes in a cluster $B_i$ as its cluster size, denoted by $s(B_i)$. For a Tanner graph with bounded variable node degrees, following the above theorem, the time complexity of the cluster decomposition post-processing can be expressed as 
\begin{equation}
O\left(\sum_{i=1}^t s(B_i)^3\right) + O(n).
\label{eq:complexity}
\end{equation}
This complexity is lower than directly applying Gaussian Elimination. However, the worst-case complexity remains $O(n^3)$, similar to that of the Gaussian Elimination decoder, which occurs when there are clusters with sizes linear in $n$.

\subsection{Cluster Decomposition with Constrained Cluster Sizes}

To improve the worst-case complexity of the cluster decoder, we can modify the cluster decomposition by imposing a size constraint on the clusters. 
The process is as follows:
After constructing the cluster forest for the residual Tanner graph, we check whether any cluster exceeds a predefined size $C$. 
If all the obtained clusters have sizes less than or equal to $C$, we proceed to construct the cluster solutions recursively. 
However, if any cluster has a size larger than $C$, we terminate the decoding process and directly call it a failure. 
This modified cluster decomposition post-processing procedure is summarized in Algorithm~\ref{alg:cluster_decomp_constrained}.

\begin{algorithm}[t]
\caption{Cluster decomposition post-processing with constrained cluster sizes}
\label{alg:cluster_decomp_constrained}
\KwIn{
block length $n$, 
residual Tanner graph $\mathcal{G}'$, 
residual syndrome $s'_x$,
cluster size constraint $C$}
\KwOut{$\widehat{x}'$}
\tcp{Cluster decomposition}
Decompose $\mathcal{G}'$ into a cluster forest with Hopcroft-Tarjan's algorithm\\
\tcp{Check the cluster sizes}
\If{exists a cluster with size larger than $C$}{
\Return{decoding failure}
}
\tcp{Construct cluster solutions}
\For{each cluster tree in the forest}{
    $B_r\leftarrow$ root cluster\\
    RecursiveCompute($B_r$, $s'_x$)\\
    RecursiveSelect($B_r$, $s'_x$)
}
Merge cluster solutions into a global solution $\widehat{x}'$\\
\Return{$\widehat{x}'$}
\end{algorithm}

Next, we show that the worst-case complexity of our post-processing procedure becomes linear if we impose a constant size constraint $C$ on the clusters.
\begin{theorem}
Assume that the Tanner graph $\mathcal{G} = (\mathcal{V}, \mathcal{C}, \mathcal{E})$ has a number of edges $|\mathcal{E}|$ that is linear in $n$. If we impose a size constraint $C$ on the clusters, then the time complexity of Algorithm~\ref{alg:cluster_decomp_constrained} becomes $O(nC^2)+O(n)$.
\label{thm2}
\end{theorem}
\begin{proof}
Following the proof of Theorem~\ref{thm:complexity}, the time complexity for constructing the cluster forest is linear in the block length $n$. 
Checking whether any cluster has a size greater than $C$ also requires linear time. 
Therefore, the overall complexity of the algorithm, in the case there exists a cluster with a size larger than $C$, is $O(n)$.

If all clusters have sizes less than or equal to $C$, then the time complexity of the post-processing procedure is maximized when all clusters have the largest size possible. 
In this case we have $O(n/C)$ clusters, and the complexity in equation (\ref{eq:complexity}) becomes:
\begin{align}
&O\left(\sum_{i=1}^t s(B_i)^3\right) + O(n)\\
=\ & O\left(\frac{n}{C}\cdot C^3\right) + O(n)\\
=\ & O\left(nC^2\right) + O(n)\label{eq:nC^2}
\end{align}
\end{proof}

\begin{corollary}
Assume that the Tanner graph $\mathcal{G} = (\mathcal{V}, \mathcal{C}, \mathcal{E})$ has a number of edges $|\mathcal{E}|$ that is linear in $n$. If we impose a constant size constraint $C$ on the clusters, then the time complexity of Algorithm~\ref{alg:cluster_decomp_constrained} becomes $O(n)$. 
\end{corollary}

We note that if the cluster size constraint $C$ varies as a function of the block length $n$, the decoding complexity can range between $O(n^3)$ and $O(n)$.
For example, setting $C = O(\sqrt{n})$ results in the complexity of the cluster decoder being $O(n^2)$, following Theorem~\ref{thm2}. 
This matches the complexity of the VH decoder \cite{connolly2024fast} for hypergraph product codes.

In the case of $C=O(\sqrt{n})$, the cluster decomposition with a cluster size constraint $C$ performs strictly better than the VH decomposition\footnote{Here, we refer to the standard VH decomposition with a worst-case complexity of $O(n^2)$. Note that at the end of \cite{connolly2024fast}, the authors discuss an improved version of VH capable of solving certain stopping sets with cycles in the VH graph, but at the cost of increasing the worst-case complexity to $O(n^{2.5})$.} 
on hypergraph product codes.
This is because all VH graphs that can be solved by the VH decomposition correspond to trees of horizontal and vertical clusters with sizes less than or equal to $C=O(\sqrt{n})$. 
Therefore, any stopping set solvable by VH is also solvable by a cluster decomposition with $C=O(\sqrt{n})$.
Thus, the cluster decomposition post-processing can be viewed as a generalization of VH, offering a wide range of complexity based on the choice of $C$ and extending its use to arbitrary quantum LDPC codes.

\section{Simulation Results}
\subsection{Cluster Decoder Performance}
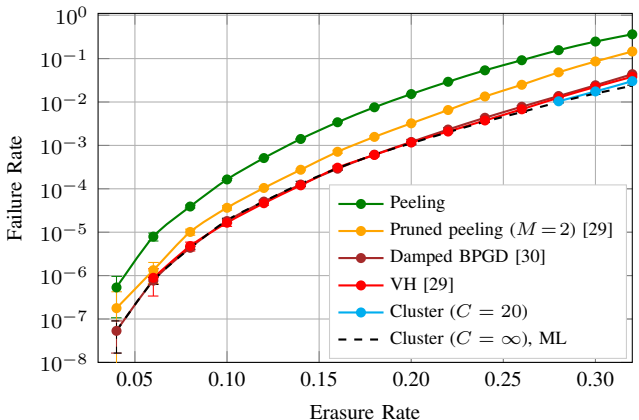
\begin{figure}[t]
\centering
\scalebox{1.0}{\begin{tikzpicture}

\definecolor{brown}{RGB}{165,42,42}
\definecolor{darkgray176}{RGB}{176,176,176}
\definecolor{green}{RGB}{0,128,0}
\definecolor{lightgray204}{RGB}{204,204,204}
\definecolor{lightgreen}{RGB}{144,238,144}
\definecolor{orange}{RGB}{255,165,0}
\definecolor{purple}{RGB}{128,0,128}
\definecolor{yellow}{RGB}{255,255,0}

\begin{axis}[
legend cell align={left},
legend style={
  fill opacity=0.7,
  draw opacity=1,
  text opacity=1,
  at={(0.99,0.01)},
  anchor=south east,
  draw=lightgray204,
  font=\scriptsize
},
width=3.8in,
height=2.7in,
xlabel={Erasure Rate},
ylabel={Failure Rate},
grid=both,
ticklabel style = {font=\footnotesize},
label style = {font=\footnotesize},
x grid style={darkgray176},
xmin=0.03, xmax=0.32,
xtick={0.00,0.05,0.10,0.15,0.20,0.25,0.30,0.35,0.40,0.45,0.50},
xticklabels={0.00,0.05,0.10,0.15,0.20,0.25,0.30,0.35,0.40,0.45,0.50},
ymode=log,
y grid style={darkgray176},
ymin=1e-8, ymax=1.10,
ytick={1,1e-1,1e-2,1e-3,1e-4,1e-05,1e-06,1e-07,1e-08},
scale=0.88
]

\addplot+[
  thick, green, mark=*, mark options={scale=0.75},
  smooth, 
  error bars/.cd, 
    y dir=both, 
    y explicit
] table [x=x, y=y, y error plus=error1, y error minus=error2, col sep=comma] {
   x,            y,       error1,       error2
0.04, 0.0000005333, 0.0000004268, 0.0000004267
0.06, 0.0000079111, 0.0000016436, 0.0000016435
0.08, 0.0000390667, 0.0000031630, 0.0000031629
0.10, 0.0001642667, 0.0000091720, 0.0000091719
0.12, 0.0005108000, 0.0000161711, 0.0000161710
0.14, 0.0014060000, 0.0000268171, 0.0000268170
0.16, 0.0034160000, 0.0001320509, 0.0001320508
0.18, 0.0075600000, 0.0001960372, 0.0001960371
0.20, 0.0152386667, 0.0002772454, 0.0002772453
0.22, 0.0292440000, 0.0003813282, 0.0003813281
0.24, 0.0537000000, 0.0011408065, 0.0011408064
0.26, 0.0917400000, 0.0014608139, 0.0014608138
0.28, 0.1565333333, 0.0058149746, 0.0058149745
0.30, 0.2472000000, 0.0069035819, 0.0069035818
0.32, 0.3628666667, 0.0076948320, 0.0076948319
};
\addlegendentry{Peeling}

\addplot+[
  thick, orange, mark=*, mark options={scale=0.75},
  smooth, 
  error bars/.cd, 
    y dir=both, 
    y explicit
] table [x=x, y=y, y error plus=error1, y error minus=error2, col sep=comma] {
   x,            y,       error1,       error2
0.04, 0.0000001778, 0.0000002464, 0.0000001777
0.06, 0.0000013333, 0.0000006748, 0.0000006747
0.08, 0.0000101333, 0.0000016110, 0.0000016109
0.10, 0.0000365333, 0.0000043258, 0.0000043257
0.12, 0.0001045333, 0.0000073170, 0.0000073169
0.14, 0.0002746667, 0.0000118596, 0.0000118595
0.16, 0.0007133333, 0.0000604250, 0.0000604249
0.18, 0.0015706667, 0.0000896244, 0.0000896243
0.20, 0.0032253333, 0.0001283250, 0.0001283249
0.22, 0.0065333333, 0.0001823348, 0.0001823347
0.24, 0.0134133333, 0.0005821654, 0.0005821653
0.26, 0.0249266667, 0.0007889713, 0.0007889712
0.28, 0.0483333333, 0.0034322289, 0.0034322288
0.30, 0.0859333333, 0.0044851834, 0.0044851833
0.32, 0.1456000000, 0.0056444515, 0.0056444514
};
\addlegendentry{Pruned peeling ($M\!=\!2$) \cite{connolly2024fast}}

\addplot+[
  thick, brown, mark=*, mark options={scale=0.75},
  smooth, 
  error bars/.cd, 
    y dir=both, 
    y explicit
] table [x=x, y=y, y error plus=error1, y error minus=error2, col sep=comma] {
   x,            y,       error1,       error2
0.04, 0.0000000533, 0.0000000370, 0.0000000369
0.06, 0.0000007667, 0.0000001401, 0.0000001400
0.08, 0.0000043667, 0.0000003344, 0.0000003343
0.10, 0.0000180667, 0.0000021510, 0.0000021509
0.12, 0.0000507333, 0.0000036045, 0.0000036044
0.14, 0.0001252667, 0.0000056637, 0.0000056636
0.16, 0.0002908000, 0.0000086287, 0.0000086286
0.18, 0.0006116216, 0.0000125960, 0.0000125959
0.20, 0.0012000000, 0.0001752025, 0.0001752024
0.22, 0.0023200000, 0.0002434726, 0.0002434725
0.24, 0.0043333333, 0.0003324134, 0.0003324133
0.26, 0.0077866667, 0.0004448246, 0.0004448245
0.28, 0.0137333333, 0.0005889733, 0.0005889732
0.30, 0.0242733333, 0.0007788238, 0.0007788237
0.32, 0.0437466667, 0.0010350692, 0.0010350691
};
\addlegendentry{Damped BPGD \cite{gokduman2024erasure}}

\addplot+[
  thick, red, mark=*, mark options={scale=0.75},
  smooth, 
  error bars/.cd, 
    y dir=both, 
    y explicit
] table [x=x, y=y, y error plus=error1, y error minus=error2, col sep=comma] {
   x,            y,       error1,       error2
0.06, 0.0000008889, 0.0000005509, 0.0000005508
0.08, 0.0000048000, 0.0000011087, 0.0000011086
0.10, 0.0000165333, 0.0000029101, 0.0000029100
0.12, 0.0000465333, 0.0000048820, 0.0000048819
0.14, 0.0001202667, 0.0000078482, 0.0000078481
0.16, 0.0003066667, 0.0000396271, 0.0000396270
0.18, 0.0006026667, 0.0000555435, 0.0000555434
0.20, 0.0011640000, 0.0000771701, 0.0000771700
0.22, 0.0020866667, 0.0001032758, 0.0001032757
0.24, 0.0037533333, 0.0003094585, 0.0003094584
0.26, 0.0067866667, 0.0004154894, 0.0004154893
0.28, 0.0126000000, 0.0017850167, 0.0017850166
0.30, 0.0223333333, 0.0023647358, 0.0023647357
0.32, 0.0382000000, 0.0030674996, 0.0030674995
};
\addlegendentry{VH \cite{connolly2024fast}}

\addplot+[
  thick, cyan, mark=*, mark options={scale=0.75},
  smooth, 
  error bars/.cd, 
    y dir=both, 
    y explicit
] table [x=x, y=y, y error plus=error1, y error minus=error2, col sep=comma] {
   x,            y,       error1,       error2
0.28, 0.0103800000, 0.0005129126, 0.0005129125
0.30, 0.0175266667, 0.0006640801, 0.0006640800
0.32, 0.0302266667, 0.0008664444, 0.0008664443
0.34, 0.0546666667, 0.0036380128, 0.0036380127
0.36, 0.1141333333, 0.0050886300, 0.0050886299
0.38, 0.2643333333, 0.0070571110, 0.0070571109
0.40, 0.6251333333, 0.0077470290, 0.0077470289
0.42, 0.9426666667, 0.0037204298, 0.0037204297
0.44, 0.9984666667, 0.0006261747, 0.0006261746
};
\addlegendentry{Cluster ($C = 20$)}

\addplot+[
  thick, black, dashed, mark=., mark options={scale=0.75},
  smooth, 
  error bars/.cd, 
    y dir=both, 
    y explicit
] table [x=x, y=y, y error plus=error1, y error minus=error2, col sep=comma] {
   x,            y,       error1,       error2
0.04, 0.0000000533, 0.0000000370, 0.0000000369
0.06, 0.0000007667, 0.0000001401, 0.0000001400
0.08, 0.0000043067, 0.0000003321, 0.0000003320
0.10, 0.0000188000, 0.0000021942, 0.0000021941
0.12, 0.0000511333, 0.0000036187, 0.0000036186
0.14, 0.0001306667, 0.0000182921, 0.0000182920
0.16, 0.0002966667, 0.0000275601, 0.0000275600
0.18, 0.0006113333, 0.0000395564, 0.0000395563
0.20, 0.0011193333, 0.0000535115, 0.0000535114
0.22, 0.0020113333, 0.0000716994, 0.0000716993
0.24, 0.0035293333, 0.0000949050, 0.0000949049
0.26, 0.0058606667, 0.0001221540, 0.0001221539
0.28, 0.0099533333, 0.0005023687, 0.0005023686
0.30, 0.0156866667, 0.0006288433, 0.0006288432
0.32, 0.0236733333, 0.0007693744, 0.0007693743
};
\addlegendentry{Cluster ($C=\infty$), ML}

\end{axis}

\end{tikzpicture}}
\caption{
Performance of the [[1600,64]] hypergraph product code over erasures with various decoders. 
The plot shows the failure rates of the decoders for recovering a Pauli-$X$ error supported on the erasure pattern, up to code degeneracy.
Error bars represent the 95\% confidence intervals around the simulated data points.
}
\label{fig:1600}
\end{figure}
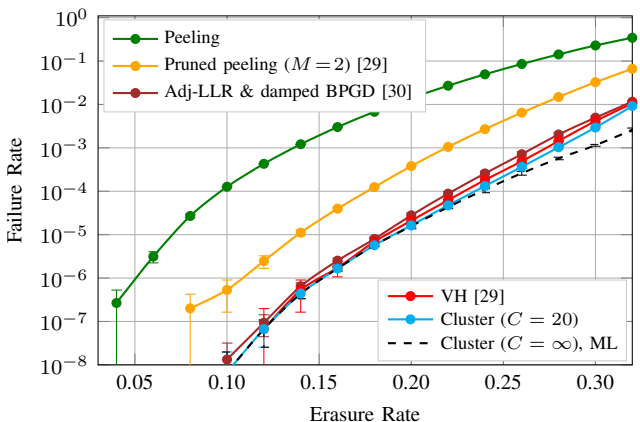
\begin{figure}[t]
\centering
\scalebox{1.0}{\begin{tikzpicture}

\definecolor{brown}{RGB}{165,42,42}
\definecolor{darkgray176}{RGB}{176,176,176}
\definecolor{green}{RGB}{0,128,0}
\definecolor{lightgray204}{RGB}{204,204,204}
\definecolor{lightgreen}{RGB}{144,238,144}
\definecolor{orange}{RGB}{255,165,0}
\definecolor{purple}{RGB}{128,0,128}
\definecolor{yellow}{RGB}{255,255,0}

\begin{axis}[
legend cell align={left},
legend style={
  fill opacity=0.7,
  draw opacity=1,
  text opacity=1,
  at={(0.01,0.99)},
  anchor=north west,
  draw=lightgray204,
  font=\scriptsize
},
width=3.8in,
height=2.7in,
xlabel={Erasure Rate},
ylabel={Failure Rate},
grid=both,
ticklabel style = {font=\footnotesize},
label style = {font=\footnotesize},
x grid style={darkgray176},
xmin=0.03, xmax=0.32,
xtick={0.00,0.05,0.10,0.15,0.20,0.25,0.30,0.35,0.40,0.45,0.50},
xticklabels={0.00,0.05,0.10,0.15,0.20,0.25,0.30,0.35,0.40,0.45,0.50},
ymode=log,
y grid style={darkgray176},
ymin=1e-8, ymax=1.10,
ytick={1,1e-1,1e-2,1e-3,1e-4,1e-05,1e-06,1e-07,1e-08},
scale=0.88
]

\addplot+[
  thick, green, mark=*, mark options={scale=0.75},
  smooth, 
  error bars/.cd, 
    y dir=both, 
    y explicit
] table [x=x, y=y, y error plus=error1, y error minus=error2, col sep=comma] {
   x,            y,       error1,       error2
0.04, 0.0000002667, 0.0000002613, 0.0000002612
0.06, 0.0000031333, 0.0000008958, 0.0000008957
0.08, 0.0000269333, 0.0000026263, 0.0000026262
0.10, 0.0001270667, 0.0000057043, 0.0000057042
0.12, 0.0004287333, 0.0000104764, 0.0000104763
0.14, 0.0012131333, 0.0000176157, 0.0000176156
0.16, 0.0030170667, 0.0000277554, 0.0000277553
0.18, 0.0067743333, 0.0000415114, 0.0000415113
0.20, 0.0139921333, 0.0000594419, 0.0000594418
0.22, 0.0270272667, 0.0000820657, 0.0000820656
0.24, 0.0494473333, 0.0003469525, 0.0003469524
0.26, 0.0855493333, 0.0004476091, 0.0004476090
0.28, 0.1422466667, 0.0005590011, 0.0005590010
0.30, 0.2289133333, 0.0021261680, 0.0021261679
0.32, 0.3438866667, 0.0024038510, 0.0024038509
};
\addlegendentry{Peeling}

\addplot+[
  thick, orange, mark=*, mark options={scale=0.75},
  smooth, 
  error bars/.cd, 
    y dir=both, 
    y explicit
] table [x=x, y=y, y error plus=error1, y error minus=error2, col sep=comma] {
   x,            y,       error1,       error2
0.08, 0.0000002000, 0.0000002263, 0.0000001999
0.10, 0.0000005333, 0.0000003696, 0.0000003695
0.12, 0.0000024667, 0.0000007948, 0.0000007947
0.14, 0.0000111333, 0.0000016886, 0.0000016885
0.16, 0.0000396667, 0.0000031872, 0.0000031871
0.18, 0.0001248000, 0.0000056532, 0.0000056531
0.20, 0.0003801333, 0.0000098650, 0.0000098649
0.22, 0.0010522000, 0.0000164071, 0.0000164070
0.24, 0.0027006667, 0.0000830536, 0.0000830535
0.26, 0.0064733333, 0.0001283406, 0.0001283405
0.28, 0.0148533333, 0.0001935855, 0.0001935854
0.30, 0.0326000000, 0.0008987156, 0.0008987155
0.32, 0.0663800000, 0.0012598364, 0.0012598363
};
\addlegendentry{Pruned peeling ($M\!=\!2$) \cite{connolly2024fast}}

\addplot+[
  thick, brown, mark=*, mark options={scale=0.75},
  smooth, 
  error bars/.cd, 
    y dir=both, 
    y explicit
] table [x=x, y=y, y error plus=error1, y error minus=error2, col sep=comma] {
   x,            y,       error1,       error2
0.10, 0.0000000133, 0.0000000185, 0.0000000132
0.12, 0.0000000933, 0.0000000489, 0.0000000488
0.14, 0.0000006467, 0.0000001287, 0.0000001286
0.16, 0.0000025200, 0.0000003593, 0.0000003592
0.18, 0.0000079667, 0.0000010100, 0.0000010099
0.20, 0.0000278667, 0.0000037780, 0.0000037779
0.22, 0.0000880000, 0.0000106149, 0.0000106148
0.24, 0.0002573333, 0.0000363009, 0.0000363008
0.26, 0.0007133333, 0.0000955403, 0.0000955402
0.28, 0.0020333333, 0.0002279674, 0.0002279673
0.30, 0.0049600000, 0.0003555261, 0.0003555260
0.32, 0.0115600000, 0.0005409592, 0.0005409591
};
\addlegendentry{Adj-LLR \& damped BPGD \cite{gokduman2024erasure}}

\addplot+[
  thick, red, mark=*, mark options={scale=0.75},
  smooth, 
  error bars/.cd, 
    y dir=both, 
    y explicit
] table [x=x, y=y, y error plus=error1, y error minus=error2, col sep=comma] {
   x,            y,       error1,       error2
0.12, 0.0000000667, 0.0000001307, 0.0000000666
0.14, 0.0000005333, 0.0000003696, 0.0000003695
0.16, 0.0000017333, 0.0000006663, 0.0000006662
0.18, 0.0000068667, 0.0000013261, 0.0000013260
0.20, 0.0000210667, 0.0000023228, 0.0000023227
0.22, 0.0000622000, 0.0000039911, 0.0000039910
0.24, 0.0001820000, 0.0000215877, 0.0000215876
0.26, 0.0004906667, 0.0000354403, 0.0000354402
0.28, 0.0014366667, 0.0000606145, 0.0000606144
0.30, 0.0041600000, 0.0003257256, 0.0003257255
0.32, 0.0105133333, 0.0005161615, 0.0005161614
};
\label{vh}

\addplot+[
  thick, cyan, mark=*, mark options={scale=0.75},
  smooth, 
  error bars/.cd, 
    y dir=both, 
    y explicit
] table [x=x, y=y, y error plus=error1, y error minus=error2, col sep=comma] {
   x,            y,       error1,       error2
0.10, 0.0000000067, 0.0000000131, 0.0000000066
0.12, 0.0000000667, 0.0000000413, 0.0000000412
0.14, 0.0000004400, 0.0000001062, 0.0000001061
0.16, 0.0000016600, 0.0000002062, 0.0000002061
0.18, 0.0000057000, 0.0000003821, 0.0000003820
0.20, 0.0000163333, 0.0000020452, 0.0000020451
0.22, 0.0000473333, 0.0000034816, 0.0000034815
0.24, 0.0001320000, 0.0000183852, 0.0000183851
0.26, 0.0003660000, 0.0000306106, 0.0000306105
0.28, 0.0010313333, 0.0000513672, 0.0000513671
0.30, 0.0029433333, 0.0000866942, 0.0000866941
0.32, 0.0092133333, 0.0004835138, 0.0004835137
};
\label{cluster_20}

\addplot+[
  thick, black, dashed, mark=., mark options={scale=0.75},
  smooth, 
  error bars/.cd, 
    y dir=both, 
    y explicit
] table [x=x, y=y, y error plus=error1, y error minus=error2, col sep=comma] {
   x,            y,       error1,       error2
0.10, 0.0000000067, 0.0000000131, 0.0000000066
0.12, 0.0000000667, 0.0000000413, 0.0000000412
0.14, 0.0000004400, 0.0000001062, 0.0000001061
0.16, 0.0000016600, 0.0000002062, 0.0000002061
0.18, 0.0000057000, 0.0000003821, 0.0000003820
0.20, 0.0000157333, 0.0000020073, 0.0000020072
0.22, 0.0000431333, 0.0000033236, 0.0000033235
0.24, 0.0001093333, 0.0000167326, 0.0000167325
0.26, 0.0002606667, 0.0000258343, 0.0000258342
0.28, 0.0005706667, 0.0000382189, 0.0000382188
0.30, 0.0011453333, 0.0000541287, 0.0000541286
0.32, 0.0026066667, 0.0002580396, 0.0002580395
};
\label{ml}

\node [
    fill=white,
    fill opacity=0.7,
    draw opacity=1,
    text opacity=1,
    anchor=south east,
    draw=lightgray204,
    font=\scriptsize,
] at (rel axis cs: 0.99,0.01) {
\shortstack[l]{
\ref{vh} VH \cite{connolly2024fast}\\
\ref{cluster_20} Cluster ($C=20$)\\
\ref{ml} Cluster ($C=\infty$), ML
}
};

\end{axis}

\end{tikzpicture}}
\caption{
Performance of the [[2025,81]] hypergraph product code over erasures with various decoders. 
The plot shows the failure rates of the decoders for recovering a Pauli-$X$ error supported on the erasure pattern, up to code degeneracy.
Error bars represent the 95\% confidence intervals around the simulated data points.
}
\label{fig:2025}
\end{figure}

In Figures \ref{fig:1600} and \ref{fig:2025}, we present simulation results for the cluster decoder for the [[1600, 64]] and [[2025, 81]] hypergraph product codes over erasures, respectively. 
These codes are two of the four hypergraph product codes simulated in \cite[Figure 4]{connolly2024fast}. 
In both figures, the blue curves represent the performance of the cluster decoder with a cluster size constraint of $C=20$, 
while the black curves show the cluster decoder performance without a cluster size constraint, equivalent to ML decoding. 
For comparison, we also include the performance of the peeling decoder, the pruned peeling decoder from \cite{connolly2024fast}, the VH decoder from \cite{connolly2024fast}, and the BPGD decoder from \cite{gokduman2024erasure}. 
The results show that, for both hypergraph product codes, 
the cluster decoder with $C=20$ closely approaches ML performance in the low-erasure-rate regime, similar to both the VH decoder and the BPGD decoder.

\begin{figure}[t]
\centering
\scalebox{1.0}{\begin{tikzpicture}

\definecolor{brown}{RGB}{165,42,42}
\definecolor{darkgray176}{RGB}{176,176,176}
\definecolor{green}{RGB}{0,128,0}
\definecolor{lightgray204}{RGB}{204,204,204}
\definecolor{lightgreen}{RGB}{144,238,144}
\definecolor{orange}{RGB}{255,165,0}
\definecolor{purple}{RGB}{128,0,128}
\definecolor{yellow}{RGB}{255,255,0}

\definecolor{t_blue}{HTML}{1f77b4}
\definecolor{t_orange}{HTML}{ff7f0e}
\definecolor{t_green}{HTML}{2ca02c}
\definecolor{t_red}{HTML}{d62728}
\definecolor{t_purple}{HTML}{9467bd}
\definecolor{t_brown}{HTML}{8c564b}
\definecolor{t_pink}{HTML}{e377c2}
\definecolor{t_gray}{HTML}{7f7f7f}
\definecolor{t_olive}{HTML}{bcbd22}
\definecolor{t_cyan}{HTML}{17becf}

\begin{axis}[
legend cell align={left},
legend style={
  fill opacity=0.7,
  draw opacity=1,
  text opacity=1,
  at={(0.01,0.99)},
  anchor=north west,
  draw=lightgray204,
  font=\scriptsize
},
width=3.8in,
height=2.7in,
xlabel={Erasure Rate},
ylabel={Failure Rate},
grid=both,
ticklabel style = {font=\footnotesize},
label style = {font=\footnotesize},
x grid style={darkgray176},
xmin=0.01, xmax=0.52,
xtick={0.00,0.05,0.10,0.15,0.20,0.25,0.30,0.35,0.40,0.45,0.50},
xticklabels={0.00,0.05,0.10,0.15,0.20,0.25,0.30,0.35,0.40,0.45,0.50},
ymode=log,
y grid style={darkgray176},
ymin=2e-9, ymax=2.0,
ytick={1,1e-1,1e-2,1e-3,1e-4,1e-05,1e-06,1e-07,1e-08},
scale=0.88
]

\addplot+[
  thick, t_green, mark=o, mark options={scale=0.75},
  smooth, 
  error bars/.cd, 
    y dir=both, 
    y explicit
] table [x=x, y=y, y error plus=error1, y error minus=error2, col sep=comma] {
   x,            y,       error1,       error2
0.28, 0.0103800000, 0.0005129126, 0.0005129125
0.30, 0.0175266667, 0.0006640801, 0.0006640800
0.32, 0.0302266667, 0.0008664444, 0.0008664443
0.34, 0.0546666667, 0.0036380128, 0.0036380127
0.36, 0.1141333333, 0.0050886300, 0.0050886299
0.38, 0.2643333333, 0.0070571110, 0.0070571109
0.40, 0.6251333333, 0.0077470290, 0.0077470289
0.42, 0.9426666667, 0.0037204298, 0.0037204297
0.44, 0.9984666667, 0.0006261747, 0.0006261746
0.46, 1.0000000000, 0.0000000000, 0.0000000000
0.48, 1.0000000000, 0.0000000000, 0.0000000000
0.50, 1.0000000000, 0.0000000000, 0.0000000000
};
\addlegendentry{[[1600,64]], $C = 20$}

\addplot+[
  thick, t_blue, solid, mark=o, mark options={scale=0.75},
  smooth, 
  error bars/.cd, 
    y dir=both, 
    y explicit
] table [x=x, y=y, y error plus=error1, y error minus=error2, col sep=comma] {
   x,            y,       error1,       error2
0.04, 0.0000000533, 0.0000000370, 0.0000000369
0.06, 0.0000007667, 0.0000001401, 0.0000001400
0.08, 0.0000043067, 0.0000003321, 0.0000003320
0.10, 0.0000188000, 0.0000021942, 0.0000021941
0.12, 0.0000511333, 0.0000036187, 0.0000036186
0.14, 0.0001306667, 0.0000182921, 0.0000182920
0.16, 0.0002966667, 0.0000275601, 0.0000275600
0.18, 0.0006113333, 0.0000395564, 0.0000395563
0.20, 0.0011193333, 0.0000535115, 0.0000535114
0.22, 0.0020113333, 0.0000716994, 0.0000716993
0.24, 0.0035293333, 0.0000949050, 0.0000949049
0.26, 0.0058606667, 0.0001221540, 0.0001221539
0.28, 0.0099533333, 0.0005023687, 0.0005023686
0.30, 0.0156866667, 0.0006288433, 0.0006288432
0.32, 0.0236733333, 0.0007693744, 0.0007693743
0.34, 0.0346000000, 0.0029248402, 0.0029248401
0.36, 0.0508000000, 0.0035141570, 0.0035141569
0.38, 0.0799333333, 0.0043399433, 0.0043399432
0.40, 0.1121333333, 0.0050495385, 0.0050495384
0.42, 0.1708000000, 0.0060226000, 0.0060225999
0.44, 0.2442000000, 0.0068752218, 0.0068752217
0.46, 0.3956000000, 0.0078252967, 0.0078252966
0.48, 0.7446000000, 0.0069788265, 0.0069788264
0.50, 0.9762666667, 0.0024359818, 0.0024359817
};
\addlegendentry{[[1600,64]], ML}

\addplot+[
  thick, t_pink, mark=o, mark options={scale=0.75},
  smooth, 
  error bars/.cd, 
    y dir=both, 
    y explicit
] table [x=x, y=y, y error plus=error1, y error minus=error2, col sep=comma] {
   x,            y,       error1,       error2
0.20, 0.0000163333, 0.0000020452, 0.0000020451
0.22, 0.0000473333, 0.0000034816, 0.0000034815
0.24, 0.0001320000, 0.0000183852, 0.0000183851
0.26, 0.0003660000, 0.0000306106, 0.0000306105
0.28, 0.0010313333, 0.0000513672, 0.0000513671
0.30, 0.0029433333, 0.0000866942, 0.0000866941
0.32, 0.0092133333, 0.0004835138, 0.0004835137
0.34, 0.0236000000, 0.0024292954, 0.0024292953
0.36, 0.0656000000, 0.0039621305, 0.0039621304
0.38, 0.1982000000, 0.0063796271, 0.0063796270
0.40, 0.5893333333, 0.0078729170, 0.0078729169
0.42, 0.9475333333, 0.0035682020, 0.0035682019
0.44, 0.9992000000, 0.0004524615, 0.0004524614
0.46, 1.0000000000, 0.0000000000, 0.0000000000
0.48, 1.0000000000, 0.0000000000, 0.0000000000
0.50, 1.0000000000, 0.0000000000, 0.0000000000
};
\label{2025_c_20}

\addplot+[
  thick, t_purple, solid, mark=o, mark options={scale=0.75},
  smooth, 
  error bars/.cd, 
    y dir=both, 
    y explicit
] table [x=x, y=y, y error plus=error1, y error minus=error2, col sep=comma] {
   x,            y,       error1,       error2
0.10, 0.0000000067, 0.0000000131, 0.0000000066
0.12, 0.0000000667, 0.0000000413, 0.0000000412
0.14, 0.0000004400, 0.0000001062, 0.0000001061
0.16, 0.0000016600, 0.0000002062, 0.0000002061
0.18, 0.0000057000, 0.0000003821, 0.0000003820
0.20, 0.0000157333, 0.0000020073, 0.0000020072
0.22, 0.0000431333, 0.0000033236, 0.0000033235
0.24, 0.0001093333, 0.0000167326, 0.0000167325
0.26, 0.0002606667, 0.0000258343, 0.0000258342
0.28, 0.0005706667, 0.0000382189, 0.0000382188
0.30, 0.0011453333, 0.0000541287, 0.0000541286
0.32, 0.0026066667, 0.0002580396, 0.0002580395
0.34, 0.0036000000, 0.0009584701, 0.0009584700
0.36, 0.0082000000, 0.0014432096, 0.0014432095
0.38, 0.0172666667, 0.0020846471, 0.0020846470
0.40, 0.0328000000, 0.0028503980, 0.0028503979
0.42, 0.0598000000, 0.0037946450, 0.0037946449
0.44, 0.1086666667, 0.0049805659, 0.0049805658
0.46, 0.2360666667, 0.0067960334, 0.0067960333
0.48, 0.6914000000, 0.0073921895, 0.0073921894
0.50, 0.9828000000, 0.0020806893, 0.0020806892
};
\label{2025_ml}

\addplot+[
  thick, t_cyan, mark=o, mark options={scale=0.75},
  smooth, 
  error bars/.cd, 
    y dir=both, 
    y explicit
] table [x=x, y=y, y error plus=error1, y error minus=error2, col sep=comma] {
   x,            y,       error1,       error2
0.20, 0.0000000347, 0.0000000136, 0.0000000135
0.22, 0.0000001533, 0.0000000362, 0.0000000361
0.24, 0.0000008444, 0.0000001342, 0.0000001341
0.26, 0.0000045810, 0.0000003136, 0.0000003135
0.28, 0.0000256136, 0.0000014954, 0.0000014953
0.30, 0.0001230000, 0.0000153698, 0.0000153697
0.32, 0.0005580000, 0.0000462862, 0.0000462861
0.34, 0.0022920000, 0.0000937271, 0.0000937270
0.36, 0.0088230000, 0.0001832904, 0.0001832903
0.38, 0.0327766667, 0.0006371493, 0.0006371492
0.40, 0.2841333333, 0.0016138851, 0.0016138850
0.42, 0.9955906040, 0.0002378909, 0.0002378908
0.44, 1.0000000000, 0.0000000000, 0.0000000000
0.46, 1.0000000000, 0.0000000000, 0.0000000000
0.48, 1.0000000000, 0.0000000000, 0.0000000000
0.50, 1.0000000000, 0.0000000000, 0.0000000000
};
\label{10000_c_20}

\addplot+[
  thick, t_olive, solid, mark=o, mark options={scale=0.75},
  smooth, 
  error bars/.cd, 
    y dir=both, 
    y explicit
] table [x=x, y=y, y error plus=error1, y error minus=error2, col sep=comma] {
   x,            y,       error1,       error2
0.20, 0.0000000153, 0.0000000090, 0.0000000089
0.22, 0.0000000222, 0.0000000138, 0.0000000137
0.24, 0.0000001556, 0.0000000576, 0.0000000575
0.26, 0.0000004134, 0.0000000942, 0.0000000941
0.28, 0.0000011830, 0.0000001424, 0.0000001423
0.30, 0.0000030667, 0.0000006267, 0.0000006266
0.32, 0.0000093333, 0.0000015461, 0.0000015460
0.34, 0.0000236000, 0.0000024584, 0.0000024583
0.36, 0.0000613103, 0.0000040302, 0.0000040301
0.38, 0.0001133333, 0.0000380934, 0.0000380933
0.40, 0.0004000000, 0.0000715548, 0.0000715547
0.42, 0.0009597315, 0.0001111768, 0.0001111767
0.44, 0.0026666667, 0.0007533978, 0.0007533977
0.46, 0.0073888889, 0.0012511200, 0.0012511199
0.48, 0.7464444444, 0.0063555791, 0.0063555790
0.50, 1.0000000000, 0.0000000000, 0.0000000000
};
\label{10000_ml}

\node [
    fill=white,
    fill opacity=0.7,
    draw opacity=1,
    text opacity=1,
    anchor=south east,
    draw=lightgray204,
    font=\scriptsize,
] at (rel axis cs: 0.99,0.01) {
\shortstack[l]{
\ref{2025_c_20} [[2025,81]], $C = 20$\\
\ref{2025_ml} [[2025,81]], ML\\
\ref{10000_c_20} [[10000,400]], $C = 20$\\
\ref{10000_ml} [[10000,400]], ML
}
};

\end{axis}
\end{tikzpicture}}
\caption{Performance of the cluster decoder on the [[1600,64]], [[2025,81]], and [[10000,400]] hypergraph product codes over erasures. 
The plot shows the failure rates of the decoders for recovering a Pauli-$X$ error supported on the erasure pattern, up to code degeneracy.
Error bars represent the 95\% confidence intervals around the simulated data points.}
\label{fig:hgp}
\end{figure}
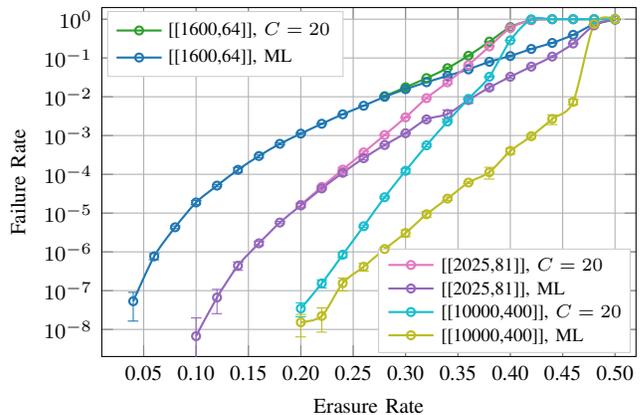

In Figure~\ref{fig:hgp}, we present the performance of the cluster decoder on three hypergraph product codes with parameters [[1600, 64]], [[2025, 81]], and [[10000, 400]] over erasures, both with a cluster size constraint of $C=20$, and without a cluster size constraint, which is equivalent to ML decoding.

For all three codes, the cluster decoder with $C=20$ approaches ML performance in the low-erasure regime, although the specific erasure rates vary. 
For the shortest [[1600, 64]] code, the cluster decoder with $C=20$ achieves close-to-ML performance when the erasure rate drops below 0.3. 
For the medium-length [[2025, 81]] code, the cluster decoder with $C=20$ closely approximates ML performance when the erasure rate drops below 0.25. 
For the longest [[10000, 400]] code, the cluster decoder with $C=20$ catches up to ML performance at an erasure rate below 0.20.

We note that the erasure rate at which the cluster decoder with a cluster size constraint $C$ approaches ML performance also depends on the value of $C$. It is expected that a larger $C$ increases the complexity of the cluster decoder but allows it to achieve close-to-ML performance at higher erasure rates.

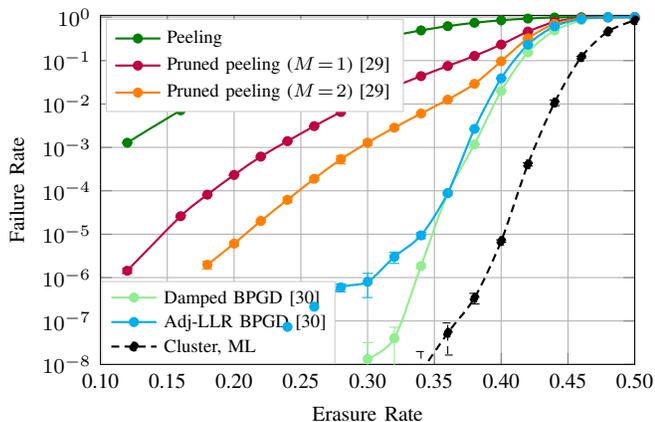
\begin{figure}[t]
\centering
\scalebox{1.0}{
\begin{tikzpicture}

\definecolor{darkgray176}{RGB}{176,176,176}
\definecolor{lightgray204}{RGB}{204,204,204}
\definecolor{lightgreen}{RGB}{144,238,144}
\definecolor{yellow}{RGB}{255,255,0}
\definecolor{green}{RGB}{0,128,0}

\begin{axis}[
legend cell align={left},
legend style={
  fill opacity=0.7,
  draw opacity=1,
  text opacity=1,
  at={(0.01,0.99)},
  anchor=north west,
  draw=lightgray204,
  font=\scriptsize
},
width=3.8in,
height=2.7in,
xlabel={Erasure Rate},
ylabel={Failure Rate},
grid=both,
ticklabel style = {font=\footnotesize},
label style = {font=\footnotesize},
x grid style={darkgray176},
xmin=0.1, xmax=0.5,
xtick={0.00,0.05,0.10,0.15,0.20,0.25,0.30,0.35,0.40,0.45,0.50},
xticklabels={0.00,0.05,0.10,0.15,0.20,0.25,0.30,0.35,0.40,0.45,0.50},
ymode=log,
y grid style={darkgray176},
ymin=1e-8, ymax=1.10,
ytick={1,1e-1,1e-2,1e-3,1e-4,1e-05,1e-06,1e-07,1e-08},
scale=0.88
]

\addplot+[
  thick, green, mark=*, mark options={scale=0.75},
  smooth, 
  error bars/.cd, 
    y dir=both, 
    y explicit
] table [x=x, y=y, y error plus=error1, y error minus=error2, col sep=comma] {
   x,            y,       error1,       error2
0.12, 0.0012803133, 0.0000057226, 0.0000057226
0.16, 0.0071674021, 0.0000237413, 0.0000237413
0.18, 0.0145025155, 0.0000336461, 0.0000336461
0.2,  0.0271782474, 0.0000457628, 0.0000457628
0.22, 0.0477916667, 0.0002414001, 0.0002414000
0.24, 0.0792880000, 0.0003057462, 0.0003057461
0.26, 0.1257083333, 0.0003751506, 0.0003751505
0.28, 0.1894600000, 0.0019831540, 0.0019831539
0.30, 0.2727400000, 0.0022538744, 0.0022538743
0.32, 0.3776133333, 0.0024533767, 0.0024533766
0.34, 0.4964466667, 0.0025302852, 0.0025302851
0.36, 0.6218933333, 0.0024540057, 0.0024540056
0.38, 0.7431066667, 0.0022111206, 0.0022111205
0.40, 0.8492066667, 0.0018109573, 0.0018109572
0.42, 0.9298000000, 0.0040885919, 0.0040885918
0.44, 0.9826000000, 0.0020925384, 0.0020925383
0.46, 0.9979333333, 0.0007267690, 0.0007267689
0.48, 1.0000000000, 0.0000000000, 0.0000000000
0.50, 1.0000000000, 0.0000000000, 0.0000000000
};
\addlegendentry{Peeling}

\addplot+[
  thick, purple, mark=*, mark options={scale=0.75},
  smooth, 
  error bars/.cd, 
    y dir=both, 
    y explicit
] table [x=x, y=y, y error plus=error1, y error minus=error2, col sep=comma] {
   x,            y,       error1,       error2
0.12, 0.0000014400, 0.0000001920, 0.0000001920 
0.16, 0.0000259793, 0.0000014345, 0.0000014345
0.18, 0.0000815463, 0.0000025414, 0.0000025414
0.2,  0.0002317113, 0.0000042836, 0.0000042836  
0.22, 0.0006136667, 0.0000280239, 0.0000280238
0.24, 0.0013926667, 0.0000422004, 0.0000422003
0.26, 0.0030890000, 0.0000627961, 0.0000627960
0.28, 0.0065466667, 0.0004081260, 0.0004081259
0.30, 0.0132266667, 0.0005781551, 0.0005781550
0.32, 0.0250733333, 0.0007912295, 0.0007912294
0.34, 0.0438866667, 0.0010366483, 0.0010366482
0.36, 0.0752066667, 0.0013346302, 0.0013346301
0.38, 0.1281133333, 0.0016913661, 0.0016913660
0.40, 0.2343866667, 0.0021437870, 0.0021437869
0.42, 0.4660666667, 0.0079832178, 0.0079832177
0.44, 0.7801333333, 0.0066278788, 0.0066278787
0.46, 0.9572000000, 0.0032391695, 0.0032391694
0.48, 0.9972666667, 0.0008355307, 0.0008355306
0.50, 1.0000000000, 0.0000000000, 0.0000000000
};
\addlegendentry{Pruned peeling ($M\!=\!1$) \cite{connolly2024fast}}

\addplot+[
  thick, orange, mark=*, mark options={scale=0.75},
  smooth, 
  error bars/.cd, 
    y dir=both, 
    y explicit
] table [x=x, y=y, y error plus=error1, y error minus=error2, col sep=comma] {
   x,            y,       error1,       error2
0.18,  0.0000019794, 0.0000003959, 0.0000003959
0.2,   0.0000060618, 0.0000006929, 0.0000006929
0.22, 0.0000201649, 0.0000012638, 0.0000012638
0.24, 0.0000620000, 0.0000089100, 0.0000089099
0.26, 0.0001890000, 0.0000155555, 0.0000155554
0.28, 0.0005333333, 0.0001168406, 0.0001168405
0.30, 0.0012866667, 0.0001814111, 0.0001814110
0.32, 0.0028533333, 0.0002699394, 0.0002699393
0.34, 0.0060200000, 0.0003914691, 0.0003914690
0.36, 0.0126133333, 0.0005647666, 0.0005647665
0.38, 0.0291400000, 0.0008512037, 0.0008512036
0.40, 0.0964000000, 0.0014936095, 0.0014936094
0.42, 0.3267333333, 0.0075058645, 0.0075058644
0.44, 0.7005333333, 0.0073299175, 0.0073299174
0.46, 0.9383333333, 0.0038495880, 0.0038495879
0.48, 0.9959333333, 0.0010184621, 0.0010184620
0.50, 1.0000000000, 0.0000000000, 0.0000000000
};
\addlegendentry{Pruned peeling ($M\!=\!2$) \cite{connolly2024fast}}


\addplot+[
  thick, lightgreen, mark=*, mark options={scale=0.75},
  smooth, 
  error bars/.cd, 
    y dir=both, 
    y explicit
] table [x=x, y=y, y error plus=error1, y error minus=error2, col sep=comma] {
   x,            y,       error1,       error2
0.30, 0.0000000133, 0.0000000185, 0.0000000132
0.32, 0.0000000400, 0.0000000320, 0.0000000319
0.34, 0.0000018533, 0.0000002179, 0.0000002178
0.36, 0.0000873333, 0.0000105747, 0.0000105746
0.38, 0.0011733333, 0.0000387393, 0.0000387392
0.40, 0.0199266667, 0.0007072241, 0.0007072240
0.42, 0.1544000000, 0.0057825125, 0.0057825124
0.44, 0.4942666667, 0.0113153213, 0.0113153212
0.46, 0.8537333333, 0.0079975919, 0.0079975918
};
\label{bpgd_d}

\addplot+[
  thick, cyan, mark=*, mark options={scale=0.75},
  smooth, 
  error bars/.cd, 
    y dir=both, 
    y explicit
] table [x=x, y=y, y error plus=error1, y error minus=error2, col sep=comma] {
   x,            y,       error1,       error2
0.24, 0.0000000733, 0.0000000433, 0.0000000432
0.26, 0.0000002133, 0.0000000739, 0.0000000738
0.28, 0.0000005933, 0.0000001233, 0.0000001232
0.30, 0.0000008000, 0.0000004526, 0.0000004525
0.32, 0.0000030000, 0.0000008765, 0.0000008764
0.34, 0.0000094000, 0.0000015516, 0.0000015515
0.36, 0.0000887333, 0.0000047669, 0.0000047668
0.38, 0.0026666667, 0.0000825306, 0.0000825305
0.40, 0.0389586667, 0.0003096589, 0.0003096588
0.42, 0.2328882985, 0.0010320290, 0.0010320289
0.44, 0.6162354197, 0.0019316578, 0.0019316577
0.46, 0.9056625415, 0.0014076207, 0.0014076206
0.48, 0.9897674419, 0.0005066418, 0.0005066417
0.50, 0.9995403279, 0.0001084376, 0.0001084375
};
\label{bpgd_a}

\addplot+[
  thick, black, mark=*, mark options={scale=0.75},
  smooth, 
  error bars/.cd, 
    y dir=both, 
    y explicit
] table [x=x, y=y, y error plus=error1, y error minus=error2, col sep=comma] {
   x,            y,       error1,       error2
0.34, 0.0000000067, 0.0000000131, 0.0000000066
0.36, 0.0000000533, 0.0000000370, 0.0000000369
0.38, 0.0000003400, 0.0000000933, 0.0000000932
0.40, 0.0000069600, 0.0000004222, 0.0000004221
0.42, 0.0004080000, 0.0000323186, 0.0000323185
0.44, 0.0108000000, 0.0016541099, 0.0016541098
0.46, 0.1204666667, 0.0052091881, 0.0052091880
0.48, 0.4699333333, 0.0079871863, 0.0079871862
0.50, 0.8470000000, 0.0057609998, 0.0057609997
};
\label{yml}

\node [
    fill=white,
    fill opacity=0.7,
    draw opacity=1,
    text opacity=1,
    anchor=south west,
    draw=lightgray204,
    font=\scriptsize,
] at (rel axis cs: 0.001,0.001) {
\shortstack[l]{
\ref{bpgd_d} Damped BPGD \cite{gokduman2024erasure}\\
\ref{bpgd_a} Adj-LLR BPGD \cite{gokduman2024erasure}\\
\ref{yml} Cluster, ML
}
};

\end{axis}

\end{tikzpicture}}
\caption{Performance of the [[882, 24]] B1 lifted-product QLDPC code from~\cite{panteleev2021degenerate} over erasures with various decoders. 
The plot shows the failure rates of the decoders for recovering a Pauli-$X$ error supported on the erasure pattern, up to code degeneracy.
Error bars represent the 95\% confidence intervals around the simulated data points.}
\label{fig:b1}
\end{figure}

In Figure~\ref{fig:b1}, we present the ML performance curve for the [[882,24]] B1 lifted-product code\footnote{This was referred to as a generalized hypergraph product code in the original paper, but the term lifted-product code is now commonly used in the literature.} from~\cite{panteleev2021degenerate} over erasures using the cluster decoder without a cluster size constraint.
Figure~\ref{fig:b1} also includes the performance of the peeling decoder, the pruned peeling decoder from \cite{connolly2024fast}, and the BPGD decoder from \cite{gokduman2024erasure}. 
Both Figure~\ref{fig:hgp} and Figure~\ref{fig:b1} demonstrate that the cluster decoder can reproduce the ML performance curve with reduced complexity for general quantum LDPC codes at high erasure rates.

\subsection{Cluster Size Analysis}

According to Theorem~\ref{thm2}, if the maximum cluster size after cluster decomposition, denoted by $s_{\max}$, is less than or equal to a constraint $C$, 
the complexity of the cluster decoder can be bounded by $O(nC^2)$.
To better understand the complexity of the cluster decoder, we provide statistics on the percentage of erasure patterns with $s_{\max}$ exceeding $C$ for different values of $C$ for the [[10000,400]] hypergraph product code. 
The performance of this code has been shown in Figure~\ref{fig:hgp}.

In Figure~\ref{fig:cluster_size}, the blue curve shows the percentage of error patterns in our simulation that are not peelable, meaning the peeling decoder alone cannot fully recover the erasures, requiring cluster decomposition post-processing. 
Among these cases, the other curves show the percentage of erasure patterns where the maximum cluster size after cluster decomposition exceeds a constraint $C$, 
with $C=10,20,50,100$ and $200$.
A sharp transition in the maximum cluster sizes is observed as the erasure rate increases from 0.38 to 0.42.
When the erasure rate is below 0.38, most stopping sets after cluster decomposition have cluster sizes less than or equal to 20. 
In contrast, when the erasure rate exceeds 0.42, nearly all erasure patterns result in maximum cluster sizes exceeding 200.

This indicates that, in most cases, the complexity of cluster decoding for the [[10000, 400]] hypergraph product code remains close to that of a linear-complexity decoder when the erasure rate is below 0.38. In this regime, nearly all erasure patterns can be peeled and decomposed into low-weight clusters, aligning with the observations in \cite{kovalev2013fault}.
However, at higher erasure rates above 0.42, the stopping set often contains a large, indecomposable cluster that requires Gaussian Elimination on the entire cluster.
For a sequence of codes with increasing length, this phenomenon is probably associated with a phase transition in the behvaior of peeling followed cluster decomposition.

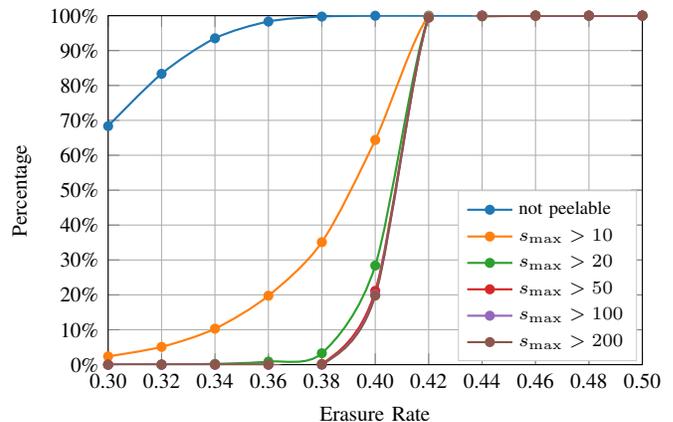
\begin{figure}[t]
\centering
\scalebox{1.0}{\begin{tikzpicture}

\definecolor{brown}{RGB}{165,42,42}
\definecolor{darkgray176}{RGB}{176,176,176}
\definecolor{green}{RGB}{0,128,0}
\definecolor{lightgray204}{RGB}{204,204,204}
\definecolor{lightgreen}{RGB}{144,238,144}
\definecolor{orange}{RGB}{255,165,0}
\definecolor{purple}{RGB}{128,0,128}
\definecolor{yellow}{RGB}{255,255,0}

\definecolor{t_blue}{HTML}{1f77b4}
\definecolor{t_orange}{HTML}{ff7f0e}
\definecolor{t_green}{HTML}{2ca02c}
\definecolor{t_red}{HTML}{d62728}
\definecolor{t_purple}{HTML}{9467bd}
\definecolor{t_brown}{HTML}{8c564b}
\definecolor{t_pink}{HTML}{e377c2}
\definecolor{t_gray}{HTML}{7f7f7f}
\definecolor{t_olive}{HTML}{bcbd22}
\definecolor{t_cyan}{HTML}{17becf}

\begin{axis}[
legend cell align={left},
legend style={
  fill opacity=0.7,
  draw opacity=1,
  text opacity=1,
  at={(0.99,0.01)},
  anchor=south east,
  draw=lightgray204,
  font=\scriptsize
},
width=3.8in,
height=2.7in,
xlabel={Erasure Rate},
ylabel={Percentage},
grid=both,
ticklabel style = {font=\footnotesize},
label style = {font=\footnotesize},
x grid style={darkgray176},
xmin=0.30, xmax=0.50,
xtick={0.30,0.32,0.34,0.36,0.38,0.40,0.42,0.44,0.46,0.48,0.50},
xticklabels={0.30,0.32,0.34,0.36,0.38,0.40,0.42,0.44,0.46,0.48,0.50},
y grid style={darkgray176},
ymin=0.00, ymax=1.00,
ytick={0.0,0.1,0.2,0.3,0.4,0.5,0.6,0.7,0.8,0.9,1.0},
yticklabels={0\%,10\%,20\%,30\%,40\%,50\%,60\%,70\%,80\%,90\%,100\%},
scale=0.88
]

\addplot+[
  thick, t_blue, mark=*, mark options={scale=0.75},
  smooth, 
] table [x=x, y=y, col sep=comma] {
   x,            y
0.30, 0.6835130000
0.32, 0.8333730000
0.34, 0.9351410000
0.36, 0.9831450000
0.38, 0.9974033333
0.40, 0.9997966667
0.42, 1.0000000000
0.44, 1.0000000000
0.46, 1.0000000000
0.48, 1.0000000000
0.50, 1.0000000000
};
\addlegendentry{not peelable}

\addplot+[
  thick, t_orange, mark=*, mark options={scale=0.75},
  smooth, 
] table [x=x, y=y, col sep=comma] {
   x,            y
0.30, 0.0237215000
0.32, 0.0509770000
0.34, 0.1030780000
0.36, 0.1975380000
0.38, 0.3509633333
0.40, 0.6437866667
0.42, 0.9982651007
0.44, 1.0000000000
0.46, 1.0000000000
0.48, 1.0000000000
0.50, 1.0000000000
};
\addlegendentry{$s_{\max}>10$}

\addplot+[
  thick, t_green, mark=*, mark options={scale=0.75},
  smooth, 
] table [x=x, y=y, col sep=comma] {
   x,            y
0.30, 0.0001190000
0.32, 0.0005490000
0.34, 0.0022730000
0.36, 0.0087800000
0.38, 0.0327100000
0.40, 0.2839866667
0.42, 0.9955906040
0.44, 1.0000000000
0.46, 1.0000000000
0.48, 1.0000000000
0.50, 1.0000000000
};
\addlegendentry{$s_{\max}>20$}

\addplot+[
  thick, t_red, mark=*, mark options={scale=0.75},
  smooth, 
] table [x=x, y=y, col sep=comma] {
   x,            y
0.30, 0.0000000000
0.32, 0.0000020000
0.34, 0.0000150000
0.36, 0.0002170000
0.38, 0.0021833333
0.40, 0.2114333333
0.42, 0.9947348993
0.44, 1.0000000000
0.46, 1.0000000000
0.48, 1.0000000000
0.50, 1.0000000000
};
\addlegendentry{$s_{\max}>50$}

\addplot+[
  thick, solid, t_purple, mark=*, mark options={scale=0.75},
  smooth, 
] table [x=x, y=y, col sep=comma] {
   x,            y
0.30, 0.0000000000
0.32, 0.0000000000
0.34, 0.0000000000
0.36, 0.0000060000
0.38, 0.0002133333
0.40, 0.2006866667
0.42, 0.9945268456
0.44, 1.0000000000
0.46, 1.0000000000
0.48, 1.0000000000
0.50, 1.0000000000
};
\addlegendentry{$s_{\max}>100$}

\addplot+[
  thick, solid, t_brown, mark=*, mark options={scale=0.75},
  smooth, 
] table [x=x, y=y, col sep=comma] {
   x,            y
0.30, 0.0000000000
0.32, 0.0000000000
0.34, 0.0000000000
0.36, 0.0000000000
0.38, 0.0000133333
0.40, 0.1978566667
0.42, 0.9944127517
0.44, 0.9977777778
0.46, 1.0000000000
0.48, 1.0000000000
0.50, 1.0000000000
};
\addlegendentry{$s_{\max}>200$}

\end{axis}

\end{tikzpicture}}
\caption{
Cluster size statistics for the [[10000, 400]] hypergraph product code.
The blue curve shows the percentage of erasure patterns that are not peelable. 
The other curves show the percentage of erasure patterns whose maximum cluster size $s_{\max}$ exceeds a constraint $C$, with $C = 10,20,50,100$ and $200$.} 
\label{fig:cluster_size}
\end{figure}

\section{Conclusion}
\label{sec:conclusion}
In this work, we introduce a new erasure decoding algorithm for quantum LDPC codes that combines classical peeling with a cluster decomposition post-processing step. 
Our approach can be viewed as a generalization of the VH decomposition method used in the previously proposed VH decoder that applies to arbitrary quantum LDPC codes. 
By either allowing unconstrained cluster sizes to achieve ML performance or constraining cluster sizes by a constant to achieve linear decoding complexity, the cluster decoder offers a versatile solution for the erasure decoding of quantum LDPC codes. 
Simulation results show that the cluster decoder with a constant cluster size can offer close-to-ML performance for hypergraph product codes in the low-erasure-rate regime. 
For the general quantum LDPC code we simulated, without the cluster size constraint, it can also effectively produce the ML performance curve at high erasure rates.

\section*{Acknowledgment}
The authors would like to thank Gregory Herschlag for helpful discussions.
This material is based on work supported by the NSF under Grants 2120757 and 2106213. 
Any opinions, findings, and conclusions or recommendations expressed in this material are those of the authors and do not necessarily reflect the views of the NSF.

\bibliographystyle{IEEEtran}
\bibliography{refs}

\begin{thebibliography}{10}
\providecommand{\url}[1]{#1}
\csname url@samestyle\endcsname
\providecommand{\newblock}{\relax}
\providecommand{\bibinfo}[2]{#2}
\providecommand{\BIBentrySTDinterwordspacing}{\spaceskip=0pt\relax}
\providecommand{\BIBentryALTinterwordstretchfactor}{4}
\providecommand{\BIBentryALTinterwordspacing}{\spaceskip=\fontdimen2\font plus
\BIBentryALTinterwordstretchfactor\fontdimen3\font minus \fontdimen4\font\relax}
\providecommand{\BIBforeignlanguage}[2]{{%
\expandafter\ifx\csname l@#1\endcsname\relax
\typeout{** WARNING: IEEEtran.bst: No hyphenation pattern has been}%
\typeout{** loaded for the language `#1'. Using the pattern for}%
\typeout{** the default language instead.}%
\else
\language=\csname l@#1\endcsname
\fi
#2}}
\providecommand{\BIBdecl}{\relax}
\BIBdecl

\bibitem{tillich2013quantum}
J.-P. Tillich and G.~Z{\'e}mor, ``Quantum {LDPC} codes with positive rate and minimum distance proportional to the square root of the blocklength,'' \emph{IEEE Transactions on Information Theory}, vol.~60, no.~2, pp. 1193--1202, 2013.

\bibitem{gottesman2013fault}
D.~Gottesman, ``Fault-tolerant quantum computation with constant overhead,'' \emph{arXiv preprint arXiv:1310.2984}, 2013.

\bibitem{fawzi2020constant}
O.~Fawzi, A.~Grospellier, and A.~Leverrier, ``Constant overhead quantum fault tolerance with quantum expander codes,'' \emph{Communications of the ACM}, vol.~64, no.~1, pp. 106--114, 2020.

\bibitem{kitaev2003fault}
\BIBentryALTinterwordspacing
A.~Kitaev, ``Fault-tolerant quantum computation by anyons,'' \emph{Annals of Physics}, vol. 303, no.~1, pp. 2--30, 2003. [Online]. Available: \url{https://www.sciencedirect.com/science/article/pii/S0003491602000180}
\BIBentrySTDinterwordspacing

\bibitem{dennis2002topological}
\BIBentryALTinterwordspacing
E.~Dennis, A.~Kitaev, A.~Landahl, and J.~Preskill, ``{Topological quantum memory},'' \emph{Journal of Mathematical Physics}, vol.~43, no.~9, pp. 4452--4505, 09 2002. [Online]. Available: \url{https://doi.org/10.1063/1.1499754}
\BIBentrySTDinterwordspacing

\bibitem{bombin2006topological}
\BIBentryALTinterwordspacing
H.~Bombin and M.~A. Martin-Delgado, ``Topological quantum distillation,'' \emph{Phys. Rev. Lett.}, vol.~97, p. 180501, Oct 2006. [Online]. Available: \url{https://link.aps.org/doi/10.1103/PhysRevLett.97.180501}
\BIBentrySTDinterwordspacing

\bibitem{panteleev2022asymptotically}
P.~Panteleev and G.~Kalachev, ``Asymptotically good quantum and locally testable classical {LDPC} codes,'' in \emph{Proceedings of the 54th Annual ACM SIGACT Symposium on Theory of Computing}, 2022, pp. 375--388.

\bibitem{leverrier2022quantum}
A.~Leverrier and G.~Z{\'e}mor, ``Quantum {Tanner} codes,'' in \emph{2022 IEEE 63rd Annual Symposium on Foundations of Computer Science (FOCS)}.\hskip 1em plus 0.5em minus 0.4em\relax IEEE, 2022, pp. 872--883.

\bibitem{dinur2023good}
I.~Dinur, M.-H. Hsieh, T.-C. Lin, and T.~Vidick, ``Good quantum {LDPC} codes with linear time decoders,'' in \emph{Proceedings of the 55th annual ACM symposium on theory of computing}, 2023, pp. 905--918.

\bibitem{bravyi2024high}
S.~Bravyi, A.~W. Cross, J.~M. Gambetta, D.~Maslov, P.~Rall, and T.~J. Yoder, ``High-threshold and low-overhead fault-tolerant quantum memory,'' \emph{Nature}, vol. 627, no. 8005, pp. 778--782, 2024.

\bibitem{bennett1997capacities}
C.~H. Bennett, D.~P. DiVincenzo, and J.~A. Smolin, ``Capacities of quantum erasure channels,'' \emph{Physical Review Letters}, vol.~78, no.~16, p. 3217, 1997.

\bibitem{wu2022erasure}
Y.~Wu, S.~Kolkowitz, S.~Puri, and J.~D. Thompson, ``Erasure conversion for fault-tolerant quantum computing in alkaline earth {Rydberg} atom arrays,'' \emph{Nature communications}, vol.~13, no.~1, p. 4657, 2022.

\bibitem{sahay2023high}
K.~Sahay, J.~Jin, J.~Claes, J.~D. Thompson, and S.~Puri, ``High-threshold codes for neutral-atom qubits with biased erasure errors,'' \emph{Physical Review X}, vol.~13, no.~4, p. 041013, 2023.

\bibitem{ma2023high}
S.~Ma, G.~Liu, P.~Peng, B.~Zhang, S.~Jandura, J.~Claes, A.~P. Burgers, G.~Pupillo, S.~Puri, and J.~D. Thompson, ``High-fidelity gates and mid-circuit erasure conversion in an atomic qubit,'' \emph{Nature}, vol. 622, no. 7982, pp. 279--284, 2023.

\bibitem{kang2023quantum}
M.~Kang, W.~C. Campbell, and K.~R. Brown, ``Quantum error correction with metastable states of trapped ions using erasure conversion,'' \emph{PRX Quantum}, vol.~4, no.~2, p. 020358, 2023.

\bibitem{kubica2023erasure}
A.~Kubica, A.~Haim, Y.~Vaknin, H.~Levine, F.~Brand{\~a}o, and A.~Retzker, ``Erasure qubits: Overcoming the {$T_1$} limit in superconducting circuits,'' \emph{Physical Review X}, vol.~13, no.~4, p. 041022, 2023.

\bibitem{teoh2023dual}
J.~D. Teoh, P.~Winkel, H.~K. Babla, B.~J. Chapman, J.~Claes, S.~J. de~Graaf, J.~W. Garmon, W.~D. Kalfus, Y.~Lu, A.~Maiti \emph{et~al.}, ``Dual-rail encoding with superconducting cavities,'' \emph{Proceedings of the National Academy of Sciences}, vol. 120, no.~41, p. e2221736120, 2023.

\bibitem{stace2009thresholds}
T.~M. Stace, S.~D. Barrett, and A.~C. Doherty, ``Thresholds for topological codes in the presence of loss,'' \emph{Physical review letters}, vol. 102, no.~20, p. 200501, 2009.

\bibitem{barrett2010fault}
S.~D. Barrett and T.~M. Stace, ``Fault tolerant quantum computation with very high threshold for loss errors,'' \emph{Physical review letters}, vol. 105, no.~20, p. 200502, 2010.

\bibitem{whiteside2014upper}
A.~C. Whiteside and A.~G. Fowler, ``Upper bound for loss in practical topological-cluster-state quantum computing,'' \emph{Physical Review A}, vol.~90, no.~5, p. 052316, 2014.

\bibitem{luby2001efficient}
M.~G. Luby, M.~Mitzenmacher, M.~A. Shokrollahi, and D.~A. Spielman, ``Efficient erasure correcting codes,'' \emph{IEEE Transactions on Information Theory}, vol.~47, no.~2, pp. 569--584, 2001.

\bibitem{richardson2001efficient}
T.~J. Richardson and R.~L. Urbanke, ``Efficient encoding of low-density parity-check codes,'' \emph{IEEE transactions on information theory}, vol.~47, no.~2, pp. 638--656, 2001.

\bibitem{delfosse2020linear}
N.~Delfosse and G.~Z{\'e}mor, ``Linear-time maximum likelihood decoding of surface codes over the quantum erasure channel,'' \emph{Physical Review Research}, vol.~2, no.~3, p. 033042, 2020.

\bibitem{delfosse2021almost}
N.~Delfosse and N.~H. Nickerson, ``Almost-linear time decoding algorithm for topological codes,'' \emph{Quantum}, vol.~5, p. 595, 2021.

\bibitem{delfosse2022toward}
N.~Delfosse, V.~Londe, and M.~E. Beverland, ``Toward a union-find decoder for quantum {LDPC} codes,'' \emph{IEEE Transactions on Information Theory}, vol.~68, no.~5, pp. 3187--3199, 2022.

\bibitem{lee2020trimming}
S.~Lee, M.~Mhalla, and V.~Savin, ``Trimming decoding of color codes over the quantum erasure channel,'' in \emph{2020 IEEE International Symposium on Information Theory (ISIT)}.\hskip 1em plus 0.5em minus 0.4em\relax IEEE, 2020, pp. 1886--1890.

\bibitem{solanki2021correcting}
H.~M. Solanki and P.~K. Sarvepalli, ``Correcting erasures with topological subsystem color codes,'' in \emph{2020 IEEE Information Theory Workshop (ITW)}.\hskip 1em plus 0.5em minus 0.4em\relax IEEE, 2021, pp. 1--5.

\bibitem{solanki2023decoding}
------, ``Decoding topological subsystem color codes over the erasure channel using gauge fixing,'' \emph{IEEE Transactions on Communications}, vol.~71, no.~7, pp. 4181--4192, 2023.

\bibitem{connolly2024fast}
N.~Connolly, V.~Londe, A.~Leverrier, and N.~Delfosse, ``Fast erasure decoder for hypergraph product codes,'' \emph{Quantum}, vol.~8, p. 1450, 2024.

\bibitem{gokduman2024erasure}
M.~G{\"o}kduman, H.~Yao, and H.~D. Pfister, ``Erasure decoding for quantum {LDPC} codes via belief propagation with guided decimation,'' in \emph{2024 60th Annual Allerton Conference on Communication, Control, and Computing}.\hskip 1em plus 0.5em minus 0.4em\relax IEEE, 2024, pp. 1--8.

\bibitem{yao2024belief}
H.~Yao, W.~A. Laban, C.~H{\"a}ger, A.~G. i~Amat, and H.~D. Pfister, ``Belief propagation decoding of quantum {LDPC} codes with guided decimation,'' in \emph{2024 IEEE International Symposium on Information Theory (ISIT)}.\hskip 1em plus 0.5em minus 0.4em\relax IEEE, 2024, pp. 2478--2483.

\bibitem{hopcroft1973algorithm}
J.~Hopcroft and R.~Tarjan, ``Algorithm 447: efficient algorithms for graph manipulation,'' \emph{Communications of the ACM}, vol.~16, no.~6, pp. 372--378, 1973.

\bibitem{cormen2022introduction}
T.~H. Cormen, C.~E. Leiserson, R.~L. Rivest, and C.~Stein, \emph{Introduction to algorithms}, 4th~ed.\hskip 1em plus 0.5em minus 0.4em\relax MIT press, 2022.

\bibitem{Harary1967}
F.~Harary and M.~D.~Plummer, ``On the core of a graph,'' \emph{Proceedings London Mathematical Society}, vol.~17, pp. 249--257, 1967.

\bibitem{gallai1964elementare}
T.~Gallai, ``Elementare relationen bez{\"u}glich der glieder und trennenden punkte von graphen,'' \emph{A Magyar Tudom{\'a}nyos Akad{\'e}mia Matematikai Kutat{\'o} Int{\'e}zet{\'e}nek k{\"o}zlem{\'e}nyei}, vol.~9, no. 1-2, pp. 235--236, 1964.

\bibitem{calderbank1996good}
\BIBentryALTinterwordspacing
A.~R. Calderbank and P.~W. Shor, ``Good quantum error-correcting codes exist,'' \emph{Phys. Rev. A}, vol.~54, pp. 1098--1105, Aug 1996. [Online]. Available: \url{https://link.aps.org/doi/10.1103/PhysRevA.54.1098}
\BIBentrySTDinterwordspacing

\bibitem{steane1996multiple}
\BIBentryALTinterwordspacing
A.~Steane, ``Multiple-particle interference and quantum error correction,'' \emph{Proceedings of the Royal Society of London. Series A: Mathematical, Physical and Engineering Sciences}, vol. 452, no. 1954, pp. 2551--2577, 1996. [Online]. Available: \url{https://royalsocietypublishing.org/doi/abs/10.1098/rspa.1996.0136}
\BIBentrySTDinterwordspacing

\bibitem{grassl1997codes}
M.~Grassl, T.~Beth, and T.~Pellizzari, ``Codes for the quantum erasure channel,'' \emph{Physical Review A}, vol.~56, no.~1, p.~33, 1997.

\bibitem{burshtein2004efficient}
D.~Burshtein and G.~Miller, ``Efficient maximum-likelihood decoding of {LDPC} codes over the binary erasure channel,'' \emph{IEEE Transactions on Information Theory}, vol.~50, no.~11, pp. 2837--2844, 2004.

\bibitem{bocharova2018ldpc}
I.~E. Bocharova, B.~D. Kudryashov, V.~Skachek, E.~Rosnes, and {\O}.~Ytrehus, ``{LDPC} codes over the {BEC}: Bounds and decoding algorithms,'' \emph{IEEE Transactions on Communications}, vol.~67, no.~3, pp. 1754--1769, 2018.

\bibitem{di2002finite}
C.~Di, D.~Proietti, I.~E. Telatar, T.~J. Richardson, and R.~L. Urbanke, ``Finite-length analysis of low-density parity-check codes on the binary erasure channel,'' \emph{IEEE Transactions on Information theory}, vol.~48, no.~6, pp. 1570--1579, 2002.

\bibitem{panteleev2021degenerate}
P.~Panteleev and G.~Kalachev, ``Degenerate quantum {LDPC} codes with good finite length performance,'' \emph{Quantum}, vol.~5, p. 585, 2021.

\bibitem{kovalev2013fault}
A.~A. Kovalev and L.~P. Pryadko, ``Fault tolerance of quantum low-density parity check codes with sublinear distance scaling,'' \emph{Physical Review A—Atomic, Molecular, and Optical Physics}, vol.~87, no.~2, p. 020304, 2013.

\end{thebibliography}

\end{document}